\documentclass[12 pt, onehalfspacing]{amsart}
\usepackage{amsthm,amsmath}
\usepackage{lineno}
\usepackage{tikz}
\usepackage{amsfonts}
\usepackage{amsmath}
\usepackage{fancyhdr}
\usepackage[T1]{fontenc}
\usepackage[latin9]{inputenc}
\usepackage{mathabx}
\usepackage{color}
\usepackage{float}
\usepackage{mathrsfs}
\usepackage{layout}
\usepackage{slashed}
\usepackage{amsthm}
\usepackage{amstext}
\usepackage{amssymb}
\usepackage{stmaryrd}
\usepackage{graphicx}
\usepackage{setspace}
\usepackage[export]{adjustbox}
\usepackage{esint}
\RequirePackage[colorlinks=true,citecolor=blue,urlcolor=blue]{hyperref}
\usepackage[authoryear]{natbib}
\usepackage{ulem}
\usepackage{enumitem}
\usepackage{ifpdf}
\usepackage{pdflscape}
\usepackage[figure]{hypcap}
\usepackage{fancyhdr}
\usepackage{bbm}
\usepackage{geometry}
\usepackage[parfill]{parskip}
\allowdisplaybreaks
\theoremstyle{plain}

\theoremstyle{definition}

\theoremstyle{plain}
\newtheorem{prop}{Proposition}
\theoremstyle{definition}
\newtheorem{defn}{Definition}
\theoremstyle{plain}
\newtheorem{defn*}{}
\theoremstyle{plain}

\theoremstyle{plain}
\newtheorem{lem}{Lemma}
\theoremstyle{plain}

\ifpdf 
 \IfFileExists{lmodern.sty}{\usepackage{lmodern}}{}
\fi 
\let\myTOC\tableofcontents
\renewcommand\tableofcontents{  \frontmatter
  \pdfbookmark[1]{\contentsname}{}
  \myTOC
  \mainmatter }

\usepackage{cleveref}

\begin{document}

		\title{Identifying Preferences when households are financially constrained}

		\author{Andreas Tryphonides\\ University of Cyprus}

		\address{P.O. Box 20537, Nicosia, Cyprus. }

		\email{tryfonidis.antreas@ucy.ac.cy}

		\date{Current version: Jan 2023, First Version  Dec 2020.}
		\thanks{\tiny{Parts of this paper draw from chapter 2 of my PhD thesis (EUI, Sep 2016), and the paper circulated under the title "Set identified Economies and Robustness to Misspecification". I thank Fabio Canova, Peter Reinhard Hansen, Frank Schorfheide and Giuseppe Ragusa for comments and suggestions on the thesis chapter. Previous versions of the paper benefited from discussions with George Marios Angeletos,  Eric Leeper, Ellen McGrattan, Juan J. Dolado, Thierry Magnac, Manuel Arellano, Raffaella Giacomini, and comments from  participants at the SED (St Louis), MMCN (Stanford), the IAAE (Thessaloniki and Nicosia), the 30th EEA Congress (Mannheim), T2M (Banque de France), the ES Winter Meetings (Barcelona), the seminars at the University of Cyprus, Royal Holloway, Humboldt University, Universidad Carlos III, University of Glasgow, University of Southampton, the Econometrics and Applied Macro working groups at the EUI. Any errors are my own.  Financial support from the BERA network, the Institute of Economic Theory II at Humboldt University and the University of Cyprus Starting Grant is gratefully acknowledged. I do not have any financial interests/personal relationships which may be considered as potential competing interests.\\\\\textit{Email}: tryfonidis.antreas@ucy.ac.cy\\\textit{Address}: Panepistimiou 1, 2109, Nicosia, Cyprus.}}

	\setcounter{page}{1}
	\pagenumbering{arabic}
	\pagestyle{plain}
	\thispagestyle{empty}

\begin{abstract}
This paper shows that utilizing information on the extensive margin of financially constrained households can narrow down the set of admissible preferences in a large class of macroeconomic models. Estimates based on Spanish aggregate data provide further empirical support for this result and suggest that accounting for this margin can bring estimates closer to microeconometric evidence. Accounting for financial constraints and the extensive margin is shown to matter for empirical asset pricing and quantifying distortions in financial markets. 
	\bigskip \newline
{\small 	Keywords:  Bounds, Financial constraints, Extensive Margin, Asset Prices  \newline JEL Classification: C51, E44, E1, E2.} 
\end{abstract}

	\maketitle

\newpage



\section{Introduction}

%

Identifying agent preferences from aggregate data has a long tradition in macroeconomics and consumption based asset pricing. One of the key challenges in recovering preferences by using the  standard representative agent Euler equation has been the fact that households are often subject to some form of liquidity constraint. The presence of occasionally binding constraints at the household level, as well as other market imperfections, also implies that agents' intertemporal marginal rates of substitution are not equalized in equilibrium, leading to non-negligible consumption dispersion that varies over time. Both of these facts render the frictionless representative agent restriction invalid.  

There is a growing literature on estimating fully specified heterogeneous agent models with specific mechanisms for restrictions on borrowing and limited risk sharing.\footnote{See for example \citet{LiuMoller,https://doi.org/10.3982/QE740,ANDM,634145,BBL,RePEc:eee:dyncon:v:115:y:2020:i:c:s0165188920300506,mongeywilliams}.} However, a natural question to ask is whether we can still obtain reasonable estimates of preferences using aggregate restrictions akin to the standard Euler equation that are consistent with the implications of financial constraints at the household level.\footnote{The motivation is similar to the representative agent literature, that is, to draw more credible inference on preferences by utilizing a subset of model predictions such as the consumption Euler equation \citep{10.2307/1912775}. Hence, identification is not tied to a particular equilibrium or a specific structure on idiosyncratic earnings risk.} This paper illustrates that in comparison to the standard representative agent framework, accounting for the presence of  constraints generates inequalities in the aggregate Euler equation. The latter also features an aggregation wedge due to imperfect risk sharing.  

The presence of weak moment inequalities leads to set identification for preferences. The set of parameter estimates that cannot be rejected directly corresponds to the set of models that generate aggregate frictions that are compatible with financial constraints. 

The key result of this paper is that information on the extensive margin of adjustment, i.e. the proportion of agents whose behavior is distorted over time due to financial frictions, can effectively narrow down the set of admissible preferences. 
The moment inequality approach is theoretically shown to be robust to under-predicting this margin in the data. This is important as defining which group of consumers is constrained or not is subject to interpretation and data availability.
 The paper provides simulation evidence that the informativeness of this margin is not compromised by measurement error, nor by the fact that cross-sectional information about consumption may be available at a lower frequency than macroeconomic aggregates. 

 The empirical application provides further support to these results by estimating preferences using Spanish macro and aggregated micro data. The proportion of constrained consumers is measured by combining information from the Survey of Household Finances (SHF) and the Business and Consumer Survey (BCS). 
 
 The paper's results on the importance of the extensive margin for sharper identification have several implications. First, 
 it contributes to our understanding that since these models are highly non-linear, certain moments can provide information about preferences without necessarily having a first order predictive power for the level of macroeconomic time series.  This is indeed the case for the borrowing constrained, as these individuals might represent a small proportion of the population, so their consumption choices  may not substantially affect the \textit{level} of aggregate consumption, yet fluctuations in this proportion are informative for preferences.  Second, had we not accounted for the extensive margin and the growth rate of cross sectional dispersion in real consumption, preference parameters would absorb the neglected variation. This can in fact help reconcile the gap between  micro and macro estimates of elasticities, contributing to the literature that studies the possible determinants of this gap.   The result speaks to both the equity premium puzzle \citep{MEHRA1985145}, in the sense that a lower degree of relative risk aversion is needed to rationalize the co-movement of consumption growth and returns, 
 and the "Frisch elasticity puzzle" \citep{10.1257/aer.101.3.471}, which concerns the gap between labor supply elasticities estimated using individual and aggregate data. 
 
 Finally, the paper argues that improved identification is not only important from an econometric point of view, but it has implications for consumption based asset pricing as well.  
 More particularly, a useful byproduct of the employed methodology is an estimate of the implied distortions to the Euler equation. In a counterfactual experiment, I show that shutting down the estimated distortions renders the  representative agent model unable to rationalize the observed equity premium since 2001, highlighting the importance of borrowing constraints and limited risk sharing. Moreover, an accounting exercise in the spirit of \citet{ECTA:ECTA768} reveals that trading frictions e.g. restrictions on short selling seem to dominate transaction costs in  distorting capital markets. The sharper estimates obtained using the extensive margin contributes to the precision of this type of inference as well.

\subsection{Contribution and contact with the literature}
The paper's methodology relates to the debate on estimating preferences using Euler equations.  \citet{ DeathtotheLogLinearizedConsumptionEulerEquationAndVeryPoorHealthtotheSecondOrderApproximation} posed a challenge as to whether the log-linear Euler equation can be consistently estimated at all, while \citet{ATTANASIO2004406} argued that by using time series variation (at the individual level) one can obtain good estimates as long as the growth of consumption variance is unrelated to the instruments used. The latter paper also finds that the estimation of non-linear Euler equations is as bad as in the case of the log-linearized Euler equation when the discount rate is too low, that is, when agents become impatient such that they hit their borrowing constraints. The methodology built in this paper contributes to tackling this issue by recognizing that the presence of borrowing constraints leads to moment inequalities instead of equalities, and that information on the probability of being constrained can sharpen identification. This holds true both in the case of using individual micro data as well as aggregate macro data. The main exposition and the application of the paper focus on the latter case, which speaks directly to the macroeconomic literature that has used the aggregate Euler equation to identify preferences. This is also less demanding in terms of data as it only requires observing households' consumption growth for a repeated cross section, while measuring the extensive margin  does not require observing the same households over time.\footnote{In order to show that moment inequalities could also be used in a microeconometric settings, Appendix \hyperref[AppB]{B} contains a fully articulated analytical example  of identification using individual time series, abstracting from unobserved heterogeneity and measurement error, which are important yet separate issue. Observing individual time series is nevertheless more demanding in terms of data availability.}

Moment inequalities have been used in the literature to characterize certain types of financial market frictions.
\citet{1996} and \citet{10.2307/2138720} derived bounds on the stochastic discount factor that are consistent with short sale constraints, proportional transaction costs and borrowing constraints but rely on complete markets and hence inequality restrictions on the representative agent Euler equation.  This paper relaxes the complete market assumption and hence the cross sectional distribution of consumption adds important identifying information. In addition, the paper focuses on preference identification and not on the non-parametric identification of stochastic discount factors \citep{doi:10.1086/261749}. Nevertheless, the informativeness of the extensive margin of constrained consumers does not hinge on the chosen class of stochastic discount factors.  
  
In the empirical literature, accounting for borrowing constraints in the aggregate Euler equation has been mostly operationalized by looking at simple forms of heterogeneity such as two agent models. For example, in recent work by \citet{ASCARI2021129}, one of the extensions of the basic representative agent Euler equation that were considered was the case of a constant fraction of hand to mouth consumers.\footnote{For a strand of the literature that follows two agent route see  \citet{RePEc:eee:dyncon:v:36:y:2012:i:11:p:1659-1672,RePEc:eee:jetheo:v:140:y:2008:i:1:p:162-196,RePEc:tpr:jeurec:v:5:y:2007:i:1:p:227-270,RePEc:pav:wpaper:124}.} This paper accommodates for a wider set of deviations that could be consistent with alternative reasons for which financial constraints break the strong correlation between consumption growth and the interest rate. Moreover, in line with the conclusions of \citet{ASCARI2021129}, the proposed methodology shows one way in which aggregated microeconomic information can be combined with traditional macroeconomic aggregates for improving identification.


Partial identification approaches have been successfully applied in a variety of  structural settings, from discrete choice in single agent models and games with multiple equilibria, to auctions and networks.\footnote{For a more comprehensive list of applications of partial identification see \citet{RePEc:ifs:cemmap:25/19} and \citet{ ho_rosen_2017}.} More specifically, \citet{ECTA:ECTA1480} have used moment inequality restrictions that arise from optimizing behavior in a single agent setting and best response behavior in multi-agent settings. Similarly, \citet{10.1093/restud/rdz007} derive moment inequalities from one-period deviations from observed export behavior of firms. The nature of these moment inequalities is nevertheless different than the ones proposed here as they arise \textit{due} to revealed preference over counterfactual choices and not due to some form of market imperfection. In the structural macroeconometric literature, moment inequalities have been used in different contexts. More particularly, \citet{https://doi.org/10.3982/QE978} treat sign restrictions in structural vector autoregressions as moment inequalities, while \citet{coroneo2011testing} use moment inequalities to test for optimal monetary policy. The objective is nevertheless the same, that is, to make inference robust to certain modeling choices.

Existing work in the microeconometric literature has also tried to disentangle identification of preferences from the particular formulation of risk sharing possibilities, but differ from this paper in terms of the restrictions they impose and the type of information they require. In order to obtain analytical tractability and point identification, \citet{10.1257/aer.104.7.2075} construct an equilibrium in which agents are either perfectly insured against some risks or not insured at all, leading to an inoperative bond market, extending the work of \cite{RePEc:ucp:jpolec:v:104:y:1996:i:2:p:219-40}. This paper differs by not imposing a particular equilibrium while its exploits aggregate information for identification. Similarly, this paper differs from  \citet{10.1257/aer.98.5.1887, 10.1257/aer.20121549} as it accounts for the possibility of binding liquidity constraints.


The rest of the paper is organized as follows. Section 2 introduces the class of models under consideration and the identification conditions. Section 3 provides identification analysis and simulation evidence, Section 4 presents the empirical application and Section 5 concludes. Appendix \hyperref[AppA]{A} contains the following material: Benchmark model extensions that account for multiple assets and transaction costs, as well as identifying conditions that account for some forms of ex-ante heterogeneity, auxiliary lemmas and proofs, additional simulation evidence on the performance of the proposed approach, the empirical approach to extract the proportion of constrained consumers, and details on the data. Appendix \hyperref[AppB]{B} contains an extended analysis of identification based on the moment inequality framework that can be applied to microeconomic data and shows analogous conditions under which the extensive margin is informative in that case as well.

\section{Model Setup}

This section outlines a class of models that is going to form the basis of the proposed identifying restrictions. A direct extension that incorporates multiple assets and transaction costs and leads to the same type of restrictions can be found in Appendix \hyperref[AppA]{A}. This Appendix also contains identifying restrictions that cater for certain forms of heterogeneity in  preferences and expectations.

\subsection{Household Optimization}
The model features households maximizing expected discounted lifetime utility subject to  an intertemporal budget constraint. Each household receives total income $y_{i,t}(l_{i,t})$ which may depend on labor hours supplied ($l_{i,t}$) and makes consumption $(c_{i,t})$, labor supply and investment decisions.\footnote{Notice also that we do not specify a specific process for labor income risk. It can therefore include any kind of individual specific risk to income as well as any kind of transfers.} The household chooses to store wealth in asset  $a_{i,t}$ which earns gross returns equal to $R_{t}=1+r(a_{i,t},S_{t})$, where $S_{t}$ is the vector of aggregate states. 
Furthermore, there is a general constraint technology $G_{i}(.)$ that restricts trades of this asset, which is increasing in the level of wealth.\footnote{This includes popular restrictions e.g. on short selling and non-collateralized borrowing  (e.g. $a_{t+1}\geq -\underline{a}$). The multiple asset case in Appendix \hyperref[AppA]{A} can accommodate more general types of constraints such as collateralized borrowing (e.g. $ a^{j}_{t+1}\geq - \sum^{J}_{k\neq j}w_{k} a^{k}_{t+1}$ or $a^{j}_{t+1}\geq - \sum^{J}_{k\neq j}w_{k} a^{k}_{t}$) where $w_{k}$ summarize restrictions on quantities across assets $j=1..J$.} 

Let $\beta\in (0,1)$ be the discount factor and $U(c_{i,t},l_{i,t};\omega,S_{t})$ the instantaneous utility function, where $\omega$ signifies parameters that are related to preferences (e.g. risk aversion, labor dis-utility).\footnote{Aggregate states could also directly affect utility. For example,  with external habits past aggregate consumption, $C_{t-1}$, belongs to $S_{t}$.} The  household problem is as follows: \vspace{-0.1 in}
\begin{eqnarray*}
	&&\max_{\{c_{i,t},l_{i,t},a_{i,t+1}\}}\mathbb{E}_{0}\sum_{t=0}^{\infty }\beta
	^{t}U(c_{i,t},l_{i,t};\omega,S_{t})\\
	s.t. &  c_{i,t} +a_{i,t+1}   &=  y_{i,t}(l_{i,t}) +\left(1+r(a_{i,t},S_{t})\right)a_{i,t}  \\ 
	& c_{i,t}>0&\text{, } G_{i}\left(a_{i,t+1},S_{t},S_{t+1}\right)\geq 0   \label{eq:model}
\end{eqnarray*} 
Denoting the marginal utility of consumption by $U'_{1}(c_{i,t};l_{i,t},\omega,S_{t})$, 
 the Euler equation is distorted by the non-negative Lagrange multiplier on the
occasionally binding constraint $(\mu_{i,t})$ which is positive when the constraint is binding:\footnote{We denote expectations conditional on information at time $t$ by $\mathbb{E}_{t}$.}
\footnotesize
\begin{eqnarray}
 U'_{1}(c_{i,t},l_{i,t};\omega,S_{t}) & = & \beta \mathbb{E}_{t}\left(1+\underset{}{r(a_{i,t+1},S_{t+1}) +\frac{\partial r(a_{i,t+1},S_{t+1})}{\partial a_{i,t+1}}a_{i,t+1}}\right)U'_{1}(c_{i,t+1},l_{i,t+1};\omega,S_{t+1}) \nonumber\\  && +\mu_{i,t}\frac{\partial G_{i}\left(a_{i,t+1},S_{t},S_{t+1}\right)}{\partial a_{i,t+1}}\nonumber 
\end{eqnarray}\normalsize
This can be simplified to  \small
\begin{eqnarray}
\label{eq:genzeldes} 
U'_{1}(c_{i,t},l_{i,t};\omega,S_{t})& = &
\beta \mathbb{E}_{t}\left(1+r^{e}(a_{i,t+1},S_{t+1})\right)U'_{1}(c_{i,t+1},l_{i,t+1};\omega,S_{t+1})+ \lambda_{i,t} 
\end{eqnarray}
\normalsize

\normalsize where $r^{e}$ is the effective rate of return, $r^{e}:=\underset{}{r(a_{i,t+1},S_{t+1}) +\frac{\partial r(a_{i,t+1},S_{t+1})}{\partial a_{i,t+1}}a_{i,t+1}}$,  and $\lambda_{i,t}$ summarizes the distortion which is positive as $\mu_{i,t}$ is positive and $G_{i}$ is increasing in wealth.
For example, in an economy with a risk free government bond, \eqref{eq:genzeldes} becomes:
\[U'_{1}(c_{i,t};l_{i,t},\omega,S_{t})=
\beta \mathbb{E}_{t}(1+r)U'_{1}(c_{i,t+1};l_{i,t+1},\omega,S_{t+1})\nonumber + \lambda_{i,t} \]  
Similar optimality conditions hold in the case in which the household makes its optimal portfolio choice in the multiple asset case.\footnote{Correspondingly, the optimality condition for hours is $\frac{\partial y_{i,t}(l_{i,t})}{\partial l_{i,t}}U'_{1}(c_{i,t},l_{i,t};\omega,S_{t})=U'_{2}(c_{i,t},l_{i,t};\omega,S_{t})$.}
As the paper focuses on the use of aggregate data for identification, I next derive the corresponding restrictions in terms of aggregates, which are directly comparable to the representative agent Euler equation. 
\subsection{Aggregation}
In order to be explicit about aggregation, let $s_{i,t}\equiv\left(y_{i,t},a_{i,t}\right)$ be the vector that includes all idiosyncratic states and $S_{t}$ the vector that includes all aggregate states, including the distribution of $s_{i,t}$,  $p(.|S_{t})$.\footnote{In a multiple-asset economy, $s_{i,t}$ would include all assets i.e. $s_{i,t}\equiv\left(y_{i,t},\left\{a^{j}_{i,t}\right\}_{j=1}^{J}\right)$.}
Dividing and multiplying individual marginal utility by the corresponding aggregate marginal utility yields that
\small
\begin{eqnarray}
	\Xi_{i,t} U'_{1}(C_{t},L_{t};\omega,S_{t})& = &
	\beta\mathbb{E}_{t}\left(1+r^{e}(a_{i,t+1},S_{t+1})\right){\Xi_{i,t+1}}{}U'_{1}(C_{t+1},L_{t+1};\omega,S_{t+1}) + {\lambda}_{i,t} 
\end{eqnarray}
\normalsize
where $\Xi_{i,t}:=\frac{U'_{1}(c_{i,t},l_{i,t};\omega,S_{t})}{U'_{1}(C_{t},L_{t};\omega,S_{t})}$ is the aggregation residual. 
Aggregating the Euler equation using $p(s_{i,t}|S_{t})$, and imposing that agent expectations are formed using $p(s_{i,t+1},S_{t+1}|s_{i,t},S_{t})$, 
\small
\begin{eqnarray}
	&& \int \Xi_{i,t}p(s_{i,t}|S_{t})d s_{i,t}  U'_{1}(C_{t},L_{t};\omega,S_{t}) \nonumber \\
	& = &
	\beta\mathbb{E}_{t}\left[\int\left(1+r^{e}(a_{i,t+1},S_{t+1})\right){\Xi_{i,t+1}}p(s_{i,t+1},s_{i,t}|S_{t+1},S_{t})d (s_{i,t+1},s_{i,t})U'_{1}(C_{t+1},L_{t+1};\omega,S_{t})\right] \nonumber \\
	&&	+ \int {\lambda}_{i,t}p(s_{i,t}|S_{t})d s_{i,t}  \label{eq:aggeul}
\end{eqnarray} \normalsize More compactly, \eqref{eq:aggeul} may be reexpressed as follows: \small
\begin{eqnarray}
	&& \Xi_{t} U'_{1}(C_{t},L_{t};\omega,S_{t}) \nonumber \\
	& = &
	\beta\mathbb{E}_{t}\left[\left(\left(1+r^{e}_{t+1}\right){\Xi_{t+1}} + \mathbb{C}ov_{t}(r^{e}_{i,t+1},\Xi_{i,t+1})\right)U'_{1}(C_{t+1},L_{t+1};\omega,S_{t+1})\right]	+ \mu_{t} \label{eq:geneul}
\end{eqnarray}
\normalsize
where  $\Xi_{t}\equiv\int \Xi_{i,t}p(s_{i,t}|S_{t})d s_{i,t}$, $\Xi_{t+1}\equiv\int \Xi_{i,t+1}p(s_{i,t+1},s_{i,t}|S_{t+1},S_{t})d  (s_{i,t+1},s_{i,t})$, \\ $\mu_{t}\equiv \int \lambda_{i,t} p(s_{i,t}|S_{t})d s_{i,t} $ and $r^{e}_{t+1}:=\int  r^{e}_{i,t+1}p(s_{i,t+1},s_{i,t}|S_{t+1},S_{t})d (s_{i,t+1},s_{i,t})$, with  $r^{e}_{i,t+1}\equiv r^{e}(a_{i,t+1},S_{t+1})$. Moreover, the cross sectional covariance between asset returns and the aggregation residual, $\mathbb{C}ov_{t}(r^{e}_{i,t+1},\Xi_{i,t+1})$, is computed using $p(s_{i,t+1},s_{i,t}|S_{t+1},S_{t})$.
\normalsize
Hence, the aggregate Euler equation  is generally distorted by dispersion in marginal utilities, its covariance with  asset returns and the aggregate multiplier.

 \section{Preference Identification and the Extensive Margin}
  Given the model setup and the corresponding individual and aggregate conditions, in this section I focus on what can be learned from the data in terms of the preference parameters $(\beta,\omega)$.  There are alternative ways of achieving identification. Splitting the sample into unconstrained and constrained households using some criterion has been predominantly used in studies employing cross sectional or panel data. Using just the unconstrained households amounts to setting $\lambda_{i,t}=0$ in \eqref{eq:genzeldes}. \citet{Zeldes1989} splits the sample using a cut-off rule on the end  of period nonhuman wealth and discusses different rules that could have been used, as well as liquidity constraints arising from credit market imperfections and illiquid assets as e.g. in \citet{doi:10.2307/1884163}. The importance of illiquid assets has been recently revived by the heterogeneous agent literature, see for example \citet{RePEc:bin:bpeajo:v:45:y:2014:i:2014-01:p:77-153} and references therein, which also identifies groups of wealthy and poor hand to mouth consumers using different rules. Similar reasoning, in a slightly different setting,  is followed 
 by \cite{RePEc:nbr:nberwo:26032} which identifies the discount factor wedge of the unconstrained households by looking at high income-low net worth households.
 
  Alternative sample-selection choices can nevertheless lead to whimsical conclusions on whether liquidity constraints are present.\footnote{Even if the selection rule identifies unconstrained individuals without error, another issue is that the moment conditions imply that the expectation error averages to zero over time, and not necessarily over the cross section (see e.g. \citet{ATTANASIO2004406} and references therein). This renders the use of cross sectional variation for preference identification less desirable. As  discussed in \citet{10.2307/2938190}, this is less of an issue when markets are complete, which is not true in our setting.}   Moreover, this approach cannot be readily applied to time series data, whether individual or aggregate.
  A more credible way to identify $(\beta,\omega)$ is to accommodate that $\lambda_{i,t}$, and therefore $\mu_{t}$, are weakly positive, that is, there is a weakly positive discrepancy between current and next period marginal utility of consumption, as households cannot smooth consumption as much as they desire. As we already alluded to, this is a  general implication of the trading constraints which is valid across models with different types of insurance mechanisms. 
 
 Conditions in \eqref{eq:genzeldes} and \eqref{eq:geneul} therefore generate conditional moment inequalities.  
  Focusing on the aggregate restriction, it follows that using any instrument that is positive i.e. an observed aggregate random variable  $(X_{t})$ that belongs to the information set of the agents, gives rise to an unconditional moment inequality: 
 \small
 \begin{eqnarray}
 	\mathbb{E}\left(1-\beta\left[\left(1+r^{e}_{t+1}\right)\frac{\Xi_{t+1}}{\Xi_{t}} + \frac{1}{\Xi_{t}}\mathbb{C}ov_{t}\left(r^{e}_{i,t+1},\Xi_{i,t+1}\right)\right]\frac{U'_{1}(C_{t+1},L_{t+1};\omega,S_{t+1})}{U'_{1}(C_{t},L_{t};\omega,S_{t})}\right)X_{t}&\geq& 0 \label{eq:geneul1}
 \end{eqnarray}
 \normalsize
The inequality does not pin down a unique vector of values for $(\beta,\omega)$, however, it implicitly defines bounds on admissible values for these parameters.\footnote{The corresponding Euler equation at the household level, interacted with a positively valued instrument $(z_{i,t})$, would yield the following unconditional moment inequality:
 \begin{eqnarray*}
 	\mathbb{E}\left[(U'_{1}(c_{i,t};l_{i,t},\omega,S_{t})-\beta(1+r^{e}_{i,t+1})U'_{1}(c_{i,t+1};l_{i,t+1},\omega, S_{t+1})) z_{i,t}\right] & \geq & 0 
 \end{eqnarray*} 
 For illustration purposes, in Appendix \hyperref[AppB]{B} I derive explicit identification regions for risk aversion $\gamma\in \omega$ using time series observations for household $i$'s consumption and assuming a constant real interest rate.}

Furthermore, implementing this identifying condition requires the choice of a particular asset.   We can in principle use any asset class, conditional on being able to measure the relevant components of \eqref{eq:geneul1}.   
  For some assets, this may be a less trivial task as  it can require assumptions on the structure of the return of the specific asset and additional cross sectional moments that will have to be identified from the data. 
 For example, the aggregate Euler equation that corresponds to a productive asset such as capital with a depreciation rate equal to $\delta$ yields the condition
 \small \begin{eqnarray*}
 	\mathbb{E}\left(1-\beta\left[\left(1-\delta + r^{k}_{t+1}\right)\frac{\Xi_{t+1}}{\Xi_{t}} + \frac{1}{\Xi_{t}}\mathbb{C}ov_{t}\left(r^{k}_{i,t+1}-\delta,\Xi_{i,t+1}\right)\right]\frac{U'_{1}(C_{t+1},L_{t+1};\omega,S_{t+1})}{U'_{1}(C_{t},L_{t};\omega,S_{t})}\right)X_{t}\geq 0 
 \end{eqnarray*} 
 \normalsize
 Empirically implementing this requires cross sectional data on firms' marginal product of capital $(r^{k}_{i})$. On the other hand,  applying the above restriction to the case of a government bond with return $r^{g}$, yields that\small
  \begin{eqnarray}
 	\mathbb{E}\left(1-\beta\left[\left(1+r^{g}_{t+1}\right)\frac{\Xi_{t+1}}{\Xi_{t}}\right]\frac{U'_{1}(C_{t+1},L_{t+1};\omega,S_{t+1})}{U'_{1}(C_{t},L_{t};\omega,S_{t})}\right) X_{t}\geq 0  \label{eq:eulerbond}
 \end{eqnarray} \normalsize which is easier to implement.
 
 \normalsize
Given these moment inequalities, I next discuss identification, and how  information on the extensive margin of adjustment for the constrained consumers can narrow down the set of admissible preferences.\footnote{Identification without parametric assumptions on the utility function is beyond the scope of the paper as the aim is to be robust with respect to the mechanism generating $\mu_{t}$ and not the primitives. In a frictionless setting, nonparametric global point identification has been proved by \citet{escanciano_hoderlein_lewbel_linton_srisuma_2020} within the class of positive marginal utility functions. For a more general class of utility functions they show that both the discount factor and the marginal utility functions are set identified. Set identification in our context has to do with the fact that the parametric Euler equation is not satisfied with equality.} 
While identification is based on the original non-linear Euler equation, some intuition can be conveyed through the lens of the log-linearized (around the steady state) condition as well.

What generates the inequality in the distorted Euler equation is the aggregate multiplier $\mu_{t}$, which is an endogenous and latent random variable.  
If $\mu_{t}$ was somehow observed over time, then \eqref{eq:geneul} would point identify preferences under standard regularity conditions.
For example, if we could measure $\mu_{t}$ in the bond Euler equation, assuming constant relative risk aversion preferences (CRRA)  with the risk aversion parameter denoted by $\gamma$, the log-linearized condition would point identify this parameter as the solution to  \small
\begin{eqnarray} \gamma&=&\frac{Cov(\tilde{R}^{g}_{t+1},\tilde{R}^{g}_{t}) + Cov(\Delta\tilde{\Xi}_{t+1}(\gamma),\tilde{R}^{g}_{t})+\frac{1-\beta \bar{R}^{g}}{\beta \bar{R}^{g}}\left(Cov(\tilde{\mu}_{t},\tilde{R}^{g}_{t})\right)}{Cov(\Delta\tilde{C}_{t+1},\tilde{R}^{g}_{t})}\label{eq:intui}
\end{eqnarray}\normalsize
where I have used the gross rate as an instrument ($\tilde{X}_{t}:=\tilde{R}^{g}_{t}$) and $\beta$ as known. 

The above expression serves to show how risk aversion is pinned down in terms of potentially observable moments. It is therefore intuitive that any information about $\mu_{t}$ is going to contribute towards identification. The fact that it can only take positive values is one kind of restriction on its domain of variation, which leads to the moment inequalities already discussed. Any improvements in identification should therefore come from additional information on $\mu_{t}$. 

One way of achieving this is to notice that the evolution of $\mu_{t}$ over time can be explained by either changes in the proportion of agents whose behavior is  distorted, the extensive margin, and by how much they distort their behavior, on average. More particularly,  $\mu_{t}$ is by definition multiplicatively decomposed  as follows: 
\begin{eqnarray*}
	\mu_{t}& =& \underset{\kappa_{t}}{\underbrace{\mathbb{E}_{t}(\lambda_{i,t}|\lambda_{i,t}>0 )}}\underset{B_{t}}{\underbrace{\mathbb{P}_{t}(\lambda_{i,t}>0)}}
\end{eqnarray*}
where $\kappa_{t}$ is the intensive margin and $B_{t}$ the extensive margin.   
These margins can always be defined for a broad range of models. If we could find a way to measure either margin, e.g. $B_{t}$, then this could be utilized to sharpen inference. The following section provides the corresponding identification analysis.

Finally, note that \eqref{eq:intui} is also useful to understand that the presence of  distortions $\Delta \tilde{\Xi}_{t}$ and $\tilde{\mu}_{t}$ can potentially explain why using the frictionless Euler equation to identify preferences could lead to substantially different estimates and conclusions. I will revisit this issue in the empirical application.

\subsection{Identification analysis} 
For ease of exposition, I briefly  lay out the identification framework based on moment inequalities, while the more technical material can be found in  Appendix \hyperref[AppA]{A} (section \ref{IdRes}). 

Since the model results in distortions that are partially observable, I generalize notation as follows: Denote the vector of observables by $Y_{t}$ and the representative agent moment function by $g\left(Y_{t},Y_{t+1},\theta\right)$ where $\theta$ are the parameters of interest. Moreover, let  $ \lambda^{o}_{t}$ be the observable component of the aggregate distortion and $\lambda^{un}_{t}$ the unobservable component. For instance, in the Euler equation that corresponds to the government bond, \eqref{eq:geneul} would correspond to  $g\left(Y_{t},Y_{t+1},\theta\right)\equiv 1- \beta \left(\frac{U'_{1}(C_{t+1},L_{t+1};\omega,S_{t+1})}{U'_{1}(C_{t},L_{t};\omega,S_{t})}\right)(1+r^{g}_{t+1})$, $\lambda^{o}_{t}\equiv  \beta \left(\frac{U'_{1}(C_{t+1},L_{t+1};\omega,S_{t+1})}{U'_{1}(C_{t},L_{t};\omega,S_{t})}\right)\left(\frac{\Xi_{t+1}}{\Xi_{t}}-1\right)(1+r^{g}_{t+1})$ and $\lambda^{un}_{t}\equiv \frac{\mu_{t}}{\Xi_{t}U'_{1}(C_{t},L_{t};\omega,S_{t})}$ where $\theta=(\beta,\omega)$.   
  Using that $\mu_{t}$ is positive, and assuming additive separability in $\lambda^{0}_{t}$, the estimating equation becomes as follows, where  $X_{t}$ is any positive random variable  which is $Y_{t}$ measurable and $\mu:=\mathbb{E}(\lambda^{un}_{t}X_{t})\in \left[ 0,\infty \right]$\footnote{Note that $[0,\infty]$ is short for $\mathbb{R}_{+}\cup\{\infty\}$.}  :	\begin{eqnarray}
  \mathbb{E}\left[\left(g\left(Y_{t},Y_{t+1},\theta\right)-\lambda^{o}_{t}\right)X_{t}\right] - \mu& = & 0   \label{eq:no_surv_mom}
  \end{eqnarray}
  Furthermore, using that $B_{t}$ and $\kappa_{t}$ have the same sign:\footnote{Appendix \hyperref[AppA]{A} (Section \ref{proofs}) includes a formal proof of conditions \eqref{eq:no_surv_mom} and \eqref{eq:surv_mom}  .}
 	\begin{eqnarray}
	\mathbb{E}\left[\left(g\left(Y_{t},Y_{t+1},\theta\right)-\lambda^{o}_{t}\right)X_{t}B^{-1}_{t}\right]- \kappa& = & 0   \label{eq:surv_mom}
	\end{eqnarray}	
	where $\kappa:=\mathbb{E}(\lambda^{un}_{t}B^{-1}_{t}X_{t})\in \left[ 0,\infty \right]$. Notice that when $B_{t}$ is zero for all $t$, then \eqref{eq:surv_mom} is trivially satisfied for all $\theta$ and \eqref{eq:no_surv_mom} becomes a moment equality $(\mu=0)$. 
%

Similarly, utilizing more than one asset leads to a vector valued function $g(.)$. Conditions like \eqref{eq:no_surv_mom} generalize as follows: \small
\begin{eqnarray}
\mathbb{E}\left[\left(\underset{\mathfrak{m}\times 1}{g}\left(Y_{t},Y_{t+1},\theta\right)-{\underset{\mathfrak{m}\times 1}{\lambda_{t}^{o}}}\right)\otimes\phi
(\mathbf{Y}^{\tau})\right]&=&\mathbb{E}\left[\lambda^{un}_{t}\otimes\phi
(\mathbf{Y}^{\tau})\right] \equiv U \in \left[ 0,\infty \right]^{\mathfrak{m}}  \label{mom_comp_gmm}
\end{eqnarray}
\normalsize where $\mathbf{Y}^{\tau}$ denotes a matrix with observations of vector ${Y}$ up to time $\tau$, $\{Y_{j}\}_{j\leq \tau}$,   $\phi (.)$ is an \textit{inequality preserving} function of $\mathbf{Y}^{\tau}$, $\mathfrak{m}$ is the number of conditions and $U$ is a vector of set valued nuisance parameters that restore the moment equality.  
The identified set $(\Theta_{I})$ is thus defined by:
\begin{eqnarray*}
\Theta_{I}&:=&\left\{\theta\in \Theta: \exists U \in \left[ 0,\infty \right]^{\mathfrak{m}}  ; \mathbb{E}\left(m(Y_{t},Y_{t+1},\theta)\otimes\phi(\mathbf{Y}^{t})\right)-U=0\right\} 
\end{eqnarray*} where 
 $m(Y_{t},Y_{t+1},\theta):=\left(g\left(Y_{t},Y_{t+1},\theta\right)-\lambda^{o}_{t}\right)$ or when the extensive margin is utilized as in \eqref{eq:surv_mom}, $m(Y_{t},Y_{t+1},\theta):=\left(g\left(Y_{t},Y_{t+1},\theta\right)-\lambda^{o}_{t}\right)\otimes B^{-1}_{t}$.

\normalsize

Identification based on these moment conditions can be analyzed by first looking at sufficient information that can identify $\theta$ for a fixed value of the wedge $U$, and then how additional information  imposes restrictions on the properties of the wedge. 

 Moment functions with information on the extensive margin can always be decomposed into the component that corresponds to the original set of conditions and the component that provides additional information. For instance, when $\mathfrak{m}=n_{\theta}=1$,\small
	\begin{eqnarray}
	\left(g\left(Y_{t},Y_{t+1},\theta\right)-\lambda^{o}_{t}\right)B^{-1}_{t} & =& \underset{m_{1}}{\underbrace{g\left(Y_{t},Y_{t+1},\theta\right)-\lambda^{o}_{t}}} + \underset{m_{2}}{\underbrace{\left(g\left(Y_{t},Y_{t+1},\theta\right)-\lambda^{o}_{t}\right)(B^{-1}_{t}-1)}}\label{eq:decomp}
\end{eqnarray} \normalsize
The moment inequality that employs information on the extensive margin subsumes the moment inequality that does not employ this information.  Since employing $m_{1}$ (that is, \ref{eq:no_surv_mom}) already set identifies $\theta$, the question is whether the additional term provides more information. To show this, I exploit the additive structure of the refined moment, that is, the conditional expectation of both $m_{1}$ and $m_{2}$ is positive. This structure is similar to the one that would result from the first order conditions of over-identified GMM, had we used $m_{1}$ and $m_{2}$ to construct two separate moment inequalities. Lemma \ref{nr} and \ref{over} in Appendix \hyperref[AppA]{A} (section \ref{IdRes} ) characterize identification in such a case.
 
It is important to stress that \eqref{eq:surv_mom} is a refinement of \eqref{eq:no_surv_mom}, not an additional moment condition. Hence its informativeness does not hinge on it being used together with \eqref{eq:no_surv_mom}, but from measuring part of the latent wedge in the Euler equation. 

The next proposition  utilizes Lemmas \ref{nr} and \ref{over} and characterizes how informative $B_{t}$  can be. In a nutshell, if the extensive margin constrains the 
 stochastic properties of the wedge $\lambda^{un}_{t}$ (and hence $\mu_{t}$ and $U$), then the size of $\Theta_{I}$ is likely to be refined, in population. 
The two main assumptions made by Lemmas \ref{nr} and \ref{over} are that the moment conditions used are correctly specified as equalities when nobody is constrained, and that conditional on a specific value for the distortion $U$, $\theta$ is point identified.\footnote{For the formal definitions, please see  Appendix \hyperref[AppA]{A} (section \ref{IdRes}).}

\begin{prop}{Identification}. Define $\Theta_{I}$ as the identified set without using the extensive margin (e.g. based on \eqref{eq:no_surv_mom}) and $\Theta_{I}'$ as the identified set when \eqref{eq:surv_mom} is utilized instead of \eqref{eq:no_surv_mom}. Given Time Series observations on $\{Y_{t},B_{t}\}_{t=1..T}$:
	\begin{enumerate}
				\item $\Theta_{I}'\subset\Theta_{I}$ if $Var(B_{t})>0$.
		\item When $B_{t}\neq 0$, $\Theta_{I}$ is not a singleton (Impossibility of point identification). 
	\end{enumerate} \label{over_surv}
\end{prop} 
\begin{proof}
	See section \ref{over_surv_proof} in Appendix \hyperref[AppA]{A}
\end{proof} \vspace{-0.1 in}

The above results can be interpreted as follows.
Variation in $B_{t}$ is essential for shrinking the identified set, as it breaks the perfect correlation between $m_{1}$ and $m_{2}$. 
Intuitively, $B_{t}$ adds information as it describes the mass of financially constrained agents over time, which determine part of the latent $\mu_{t}$.\footnote{Viewed from the perspective of an individual, $B_{t}$ describes the probability that an agent chosen at random is constrained or not. The analytical example in Appendix \hyperref[AppB]{B} uses observations at the individual level and shows explicitly how the extensive margin provides additional identifying information for the risk aversion parameter. } Viewed through the lens of the log-linearized condition \eqref{eq:intui}, incorporating information on $B_{t}$ yields that
\begin{eqnarray} \gamma&=&\frac{Cov(\tilde{R}^{g}_{t+1},\tilde{R}^{g}_{t}) + Cov(\Delta\tilde{\Xi}_{t}(\gamma),\tilde{R}^{g}_{t})+\frac{1-\beta \bar{R}^{g}}{\beta \bar{R}^{g}}\left(Cov(\tilde{B}_{t},\tilde{R}^{g}_{t})+Cov(\tilde{\kappa}_{t},\tilde{R}^{g}_{t})\right)}{Cov(\Delta\tilde{C}_{t+1},\tilde{R}^{g}_{t})} \label{eq:analytical}
\end{eqnarray}
\normalsize
As $\tilde{B}_{t}$ is now observed, the covariance between the real gross return  and the extensive margin $(Cov(\tilde{B}_{t},\tilde{R}^{g}_{t}))$ is pinned down by the data. Hence, risk aversion $(\gamma)$ will only vary in the admissible range of  $Cov(\tilde{\kappa}_{t},\tilde{R}^{g}_{t})$, the corresponding covariance with the intensive margin. Moreover, absence of time variation in the proportion of agents that are constrained implies that $B_{t}$ has no additional information for the parameters of the model as  $Cov(\tilde{\kappa}_{t},\tilde{R}^{g}_{t})=Cov(\tilde{\mu}_{t},\tilde{R}^{g}_{t})$. 

Regarding the second result of Proposition 1, when nobody is constrained at all times,  $\gamma$ will be point identified as: 
\begin{eqnarray} \gamma&=&\frac{Cov(\tilde{R}^{g}_{t+1},\tilde{R}^{g}_{t})+ Cov(\Delta\tilde{\Xi}_{t}(\gamma),\tilde{R}^{g}_{t})}{Cov(\Delta\tilde{C}_{t+1},\tilde{R}^{g}_{t})}
\end{eqnarray} Otherwise, if there is a positive mass of agents that are constrained at some $t$, then risk aversion will be set identified as the non-linear Euler equation turns into an inequality.

An additional insight  from \eqref{eq:analytical} is that the extent to which the presence of liquidity constraints matters for the identification of risk aversion depends on how close $\frac{1}{\beta}$ is to the real interest rate. In the frictionless steady state, the gross interest rate is  pinned down by $\frac{1}{\beta}$. When markets are incomplete and/or agents are subject to financial constraints, precautionary motives pull the equilibrium interest rate below the discount rate. This makes estimates of risk aversion more sensitive to the lack of observing $B_{t}$ and $\kappa_{t}$. Symmetrically, observing either of them will matter most when the extent of precautionary saving is substantial.

\subsection{Robustness of Moment Inequality} 
A concern that might arise is that the Euler equation distortion might be due to agents being constrained for a variety of reasons, which are inevitably not all captured by the available measure for $B_{t}$. It is therefore re-assuring to show that the empirical approach is still valid. The most instructive way to explain this is by viewing $B_{t}$, as a latent variable which fluctuates due to consumer choices being distorted for different reasons. 

More particularly, if $\left\{A^{i}\right\}$ signifies the set of disjoint events causing agents to behave with a distorted Euler equation, and ${B}^{o}_{t}$ is the proportion of agents observed for a subset of events in $\left\{A^{i}\right\}$, then  
\[B_{t}\equiv \mathbb{P}(\lambda_{i,t}> 0) = \sum_{i} \mathbb{P}(\left\{\lambda_{i,t}> 0\right\}\cap \mathcal{A}^{i})\geq {B}^{o}_{t}\] 
We can therefore see that (i) any information on this latent variable is helpful as it partially captures variations in $B_{t}$ and (ii) fully observing $B_{t}$ is not necessary for the approach to be successful as even if ${B}^{o}_{t}\leq B_{t}$, the moment inequality restrictions are still valid:
\begin{eqnarray*}
	{q}_{t}({B}^{o}_{t})^{-1}&\geq&{q}_{t}({B}^{}_{t})^{-1}\\
	\mathbb{E}({q}_{t}({B}^{o}_{t})^{-1}X_{t})&\geq& \mathbb{E}({q}_{t}({B}^{}_{t})^{-1}X_{t})\geq 0     
\end{eqnarray*}    
where ${q}_{t}$ is the moment condition based on the model without active constraints e.g. \[{q}_{t}:=1-\beta \mathbb{E}_{t}\left(\frac{C_{t+1}}{C_{t}}\right)^{-\gamma}\left(\frac{\Xi_{t+1}}{\Xi_{t}}\right)R_{t+1}\] In the ideal case in which we can identify all events, then there will be no variation in $B_{t}$ that is unaccounted for, and therefore identified sets are likely to be as sharp as possible. 

Moreover, notice that as in the more standard moment equality case, the approach is robust to any unobserved shock that distorts the conditional Euler equation but is orthogonal to aggregate past information. If there is a shock  that distorts the Euler equation but is not additive, then it can be easily accommodated by simulating from a postulated distribution of the shock, and then averaging across the draws.
 Having said that, distorting factors that are unobserved and cannot be controlled for can introduce additional wedges that might render the inequality invalid. An analogous qualifying statement should therefore be made as in the case of moment equalities, with the exception that moment inequalities are not valid only if the direction of these wedges is of the opposite sign, in total. The set of unobserved factors for which this holds is strictly a subset of the set for which the moment equality does not hold.  

\subsection{Illustration with Simulated Data} 
This sub-section provides supportive simulation evidence on using moment inequalities to identify risk aversion and the discount factor in a controlled environment, where the data generating process is the \citet{krusellsmith} (KS hereafter) economy, the workhorse heterogeneous agent model with aggregate risk. 
Adopting the formulation in \citet{MALIAR201042, denhaanjudd}, the KS economy  features a standard capital accumulation problem under incomplete markets, where the household first order conditions are as follows: 
\begin{eqnarray*}
	c^{-\gamma}_{i,t}&=& \beta\mathbb{E}_{t}(1-\delta + r^{k}_{t+1})c^{-\gamma}_{i,t+1}+\lambda_{i,t}\\
	k_{i,t+1} &=& (1-\tau)w_{t}\bar{l}\epsilon_{i} + \nu_{t} w_{t} (1-\epsilon_{i})+(1-\delta+r^{k}_{t})k_{i,t}-c_{i,t}\\
	k_{i,t+1}&\geq& 0 
\end{eqnarray*} 
where $\epsilon_{i}\in\{0,1\}$ is the idiosyncratic employment shock, $\bar{l}$ the time endowment and $v_{t}$ the fraction of the wage $(w_{t})$ that is received as unemployment benefit. 

Perfectly competitive markets imply that the interest and wage rates are equal to the marginal products of capital and labor: \[r^{k}_{t}=Z_{t}\alpha\left(\frac{K_{t}}{\bar{l}L_{t}}\right)^{\alpha-1},  \quad w_{t}=Z_{t}(1-\alpha)\left(\frac{K_{t}}{\bar{l}L_{t}}\right)^{\alpha}\]
while the labor income tax rate $(\tau)$ is designed to be equal to $\tau = \frac{\nu_{t} u_{t}}{\bar{l}L_{t}}$. 
The only source of aggregate risk is the aggregate productivity shock $Z_{t}$ which takes two values, $\left\{Z_{low},Z_{high}\right\}$. In turn, the unemployment rate $u_{t}$ depends on the productivity shock as well and takes two possible values, $\left\{u_{low},u_{high}\right\}$, with transition matrix $P=(p_{ll},p_{lh};p_{hl},p_{hh})$. Finally, full employment is normalized to one, thus $L = 1-u$.

The aggregated Euler equation is similar to \eqref{eq:eulerbond} but now the real interest rate is pinned down by the aggregate capital stock. The following moment inequalities characterize the restrictions without and with the extensive margin respectively:
\begin{eqnarray}
	\mathbb{E}\left(1-\beta \left(\frac{C_{t+1}}{C_{t}}\right)^{^{-\gamma}}\left(\frac{\Xi_{t+1}}{\Xi_{t}}\right)(1-\delta+ r^{k}_{t+1})\right) X_{t}&\geq& 0\\ 
	\mathbb{E}\left(1-\beta \left(\frac{C_{t+1}}{C_{t}}\right)^{^{-\gamma}}\left(\frac{\Xi_{t+1}}{\Xi_{t}}\right)(1-\delta+ r^{k}_{t+1})\right) B^{-1}_{t}X_{t}&\geq& 0   \label{simul_exp1}
\end{eqnarray}
 where instrument $X_{t}$ is constructed using standardized values of aggregate consumption lagged by one and two periods.   I adopt the same parameterization and solution method as in \citet{MALIAR201042, denhaanjudd} to generate the data by simulating 1400 periods and keeping the last 200 as the pseudo-sample. The mass of the borrowing constrained is measured as the mass on the leftmost four grid points (out of 10000) for $k_{i,t}$.
 
 The parameter values are set as follows: $\beta=0.97,\gamma=1.5,\alpha=0.36,\delta=0.025,Z_{high}=1.01,Z_{low}=0.99,\nu = 0.15,\bar{l}=\frac{1}{0.9}, u_{h}=0.04,u_{l} = 0.1$. Fixing $\delta$, I employ the above moment inequalities to estimate $(\beta,\gamma)$.

Reparameterizing the inequalities to obtain equalities, (as in \eqref{mom_comp_gmm}) the parameter estimates for $(\theta,U)$ are defined by minimizing the following Continuously Updated - GMM criterion:\vspace{-0.1 in}
\begin{eqnarray}
	Q_{T}(\theta,U)&=& \frac{1}{2}(\bar{m}(Y_{t},Y_{t+1},\theta)-U)'V^{-1}(\bar{m}(Y_{t},Y_{t+1},\theta)-U) \label{eq:cugmm}
\end{eqnarray}
where $\bar{m}(Y_{t},Y_{t+1},\theta)$ is the sample counterpart of the moment conditions above with $\theta=(\beta,\gamma)$, and $V^{-1}$ is the generalized inverse of the long run variance of the moment conditions estimated using the \citet{10.2307/1913610} estimator truncated at second order. Inference is based on quantiles of 700000 draws from the quasi-posterior distribution based on the above objective function using a single block Metropolis-Hastings Markov Chain Monte Carlo (MCMC) sampler. Similar to standard quasi-Bayesian algorithms, the quasi-posterior is proportional to $\pi(\theta)e^{-TQ_{T}(\theta,U)}$ where $\pi(\theta,U)$ is the prior distribution, see e.g. \citet{doi:10.3982/ECTA14525}. Moreover,  I do not use any prior information, apart from $U>0$ (element-wise). 

Figure \ref{MC} presents the $95\%$ confidence sets for $(\beta,\gamma)$, with and without using the extensive margin when $T=200$, a conventional sample size for  quarterly data, where the dashed lines indicate the true parameter values.

Utilizing the extensive margin shrinks the confidence sets.\footnote{As a technical aside, notice that within the context of set identification, sharper confidence sets do not imply sets that are centered around the true value, as any value within the identified set is equally likely given the identifying restrictions, in population.}  Despite that the proportion of constrained agents is not, quantitatively speaking,  a variable that influences the level of aggregate consumption, its covariation with the moment conditions is very informative for the preference parameters. 

Section \ref{MCev} in Appendix \hyperref[AppA]{A}  presents Monte Carlo evidence (500 simulations) on the reduction of the size of the confidence sets that cover the true value. Beyond the ideal setting in terms of data availability, I provide simulation evidence on the good performance of the methodology in the empirically relevant case in which the proportion of constrained agents is mismeasured, as well as the case in which the consumption cross section is available at a lower frequency than the macroeconomic aggregates.  
\begin{figure}[]
	\begin{center}			
		\includegraphics[scale=0.5]{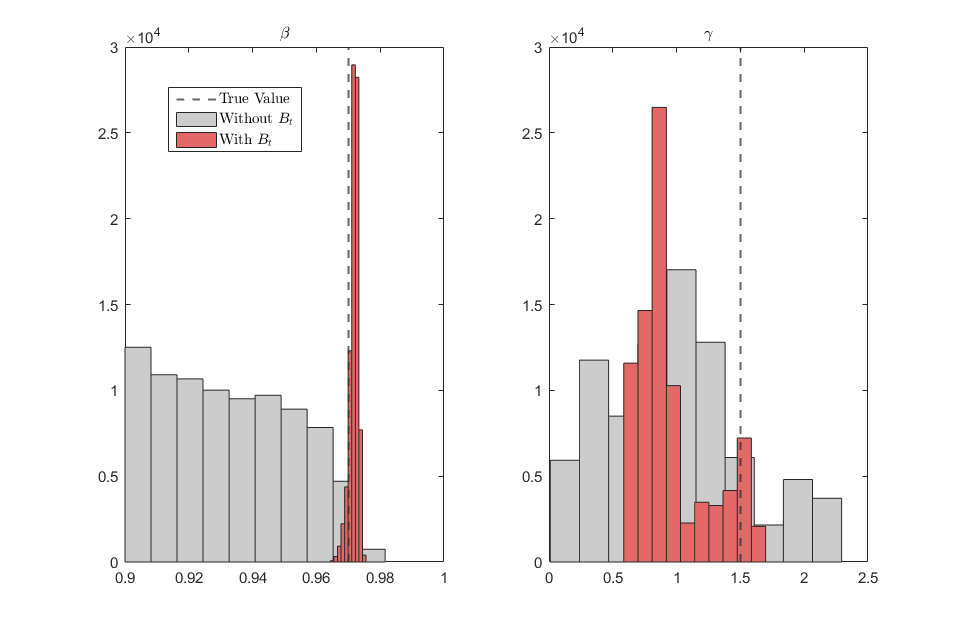}
		\par  	\label{MC}
	\end{center}  \vspace{-0.2 in}
	\protect\caption{Comparison of Confidence sets for the identified set for $T=200$  using (13) (grey) and (14) (red), respectively.}	
\end{figure} 
\subsection{A note on implementation}
The criterion function in \eqref{eq:cugmm} is a straightforward extension of the the continuously updated GMM criterion when the estimating equations are moment equalities ($U=0$). Given a conditional moment inequality, we can construct an unconditional moment inequality by multiplying the moment condition with a positively valued random variable which belongs to the information set of the agents, just as in the conventional GMM case. $\phi(\mathbf{Y}^{t})$ can therefore be any positively valued transformation of current and lagged values $Y_{t-j},j\geq 0$. In theory, interacting the distortion $\lambda^{un}_{t}$ with an instrument and taking the unconditional expectation leads to a positive scalar, the nuisance parameter $u$. For example, 
\begin{eqnarray}
	\mathbb{E}\left(1-\beta \left(\frac{C_{t+1}}{C_{t}}\right)^{^{-\gamma}}\left(\frac{\Xi_{t+1}}{\Xi_{t}}\right)(1-\delta+ r^{k}_{t+1})\right) X_{t}&\geq& 0
\end{eqnarray} is equivalent to 
\begin{eqnarray}
	\mathbb{E}\left(1-\beta \left(\frac{C_{t+1}}{C_{t}}\right)^{^{-\gamma}}\left(\frac{\Xi_{t+1}}{\Xi_{t}}\right)(1-\delta+ r^{k}_{t+1})\right) X_{t}-u&=& 0 
\end{eqnarray}
In terms of practice, what this implies is that $(\theta,u)$ have to be jointly drawn using the same MCMC procedure one would use to obtain draws from the pseudodensity induced by the GMM criterion in the moment equality case (see \citet{Chernozhukov2003293} for the latter). The difference is in how inference is done in the case of partial identification, as the one we have in this setting. Inference is based on collecting the parameter draws that correspond to the $1-\alpha$ level quantiles of the criterion draws. Under point identification, this would be equivalent to performing inference by using directly the parameter draws. Under set identification, because there can be multiple parameter draws that correspond to the same value of the criterion function, the former do not necessarily converge, while draws from the pseudo-density do. 

More particularly, applying Procedure 1 in \citet{doi:10.3982/ECTA14525}, 
\begin{enumerate}
	\item Draw a sample $\{\beta^{b},\gamma^{b},u^{b}_{1},u^{b}_{2},...u^{b}_{n_{\mathfrak{m}}}\}_{b=1..B}$ from the quasi-posterior $\pi(\theta,U)e^{-TQ_{T}(\theta,U)}$.
	\item Calculate the $(1-\alpha)$ quantile of $L_{T}:=-Q_{T}(\theta,U)$, $\zeta_{B,\alpha}$. 
	\item The $100\alpha\%$ confidence set for the identified set are the parameter draws such that $L_{T}\geq \zeta_{B,\alpha}$.
\end{enumerate}

 \citet{doi:10.3982/ECTA14525} have shown that confidence sets constructed in this way have correct frequentist coverage when the optimal weighting matrix is utilized.
Moreover, they provide algorithms to perform subvector inference that is less conservative that the projection based approach that is utilized in this paper.

%


\section{Application: Estimating Preferences using Spanish Data} 
 Spain is a country in which financial frictions are a  priori expected to have played a significant role post 1999, and especially during the crisis period, 2008-2014. Unrestrained credit growth in the past and the subsequent bust has resulted in firm bankruptcies and a large drop in investment. Similarly, household indebtedness before the crisis was succeeded by de-leveraging.

\begin{figure}[H]
	\begin{centering}
		\includegraphics[scale=0.22]{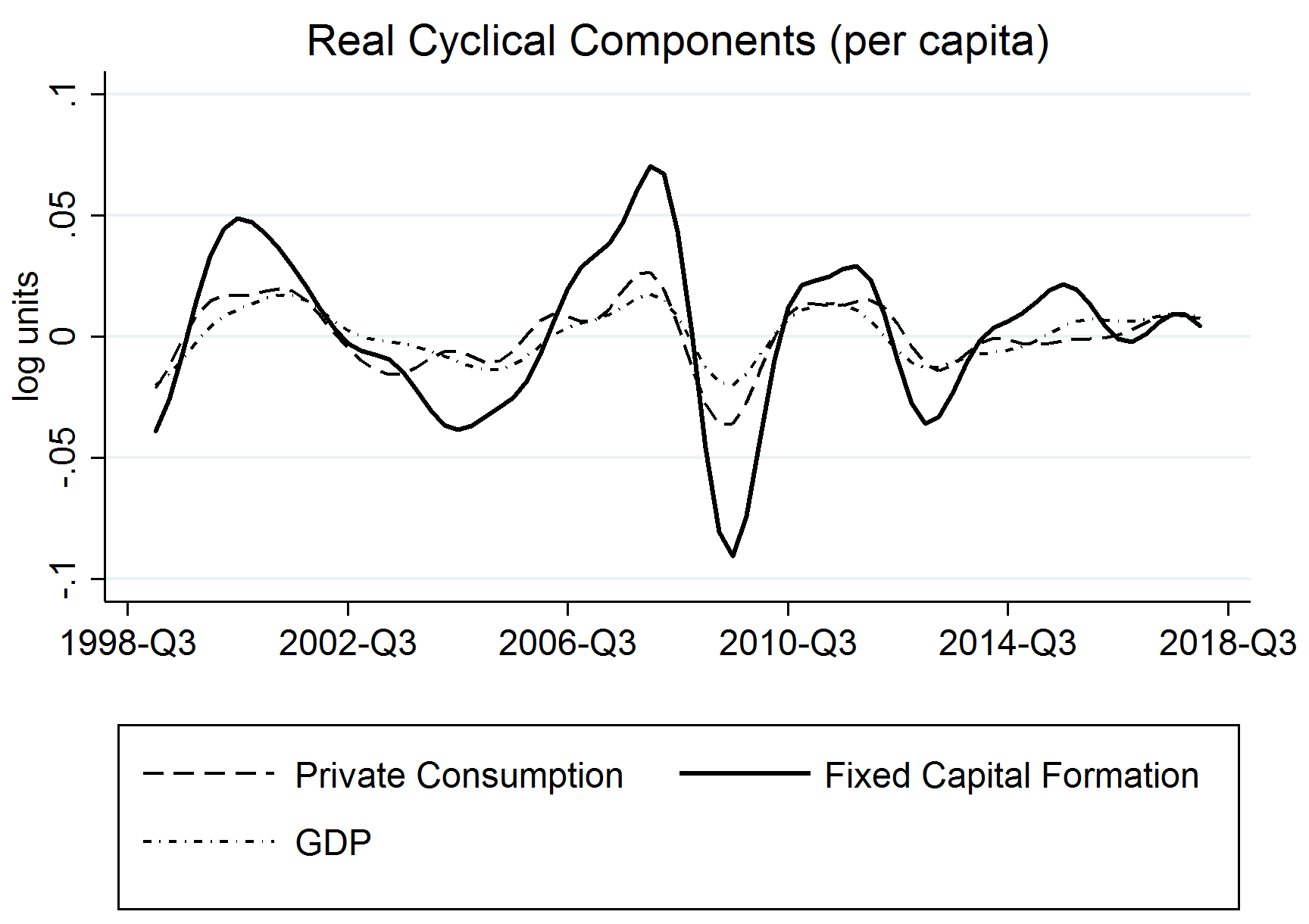}
		\par
	\end{centering} 
\caption{}
	\label{spain}
\end{figure} \vspace{-0.25 in}
Figure \ref{spain} plots the cyclical components for consumption, investment and output after 1999. 
As evident, there was a large slump in aggregate investment, followed by a relatively smaller slump in aggregate consumption.

 An advantage of the framework developed in this paper is that the identification of deep parameters is not tied to taking a stance on how severe financial frictions have been, nor the time period in which credit constraints have been binding. This allows us to use the entire time series in the post - Euro period, as well as survey data that are partially harmonized across countries of the European Area, making the approach portable to other countries as well.\footnote{Another interesting case to apply this methodology is the case of the US using e.g. data from the Survey of Consumer Finances (also available for every three years) and FRBNY microeconomic data centre (Survey of Consumer Expectations). Indeed the latter dataset is available at a high frequency, on par with macroeconomic variables, but this survey started much later than EU surveys (2013).}

\subsection*{Identifying Restrictions}
Without loss of generality with respect to the friction generating mechanism, I make parametric assumptions on the utility function.
Assuming CRRA preferences and separability between consumption and leisure implies the following utility function:
\[u(c_{i,t},l_{i,t},\gamma,\eta):= \frac{c_{i,t}^{1-\gamma }-1}{1-\gamma }- \frac{l_{i,t}^{1+\eta}}{1+\eta}   \] where $\gamma$ is the coefficient of relative risk aversion and $\eta$ the (inverse) Frisch elasticity of labor supply. 

The relevant first order conditions from \eqref{eq:genzeldes} are the Euler equation for some asset $j$, \small 
\begin{eqnarray*}
c_{i,t}^{-\gamma}&=&\beta\mathbb{E}_{t}c_{i,t+1}^{-\gamma}\left(1+r_{j}(a^{j}_{i,t+1},S_{t+1})+\frac{\partial r_{j}(a^{j}_{i,t+1},S_{t+1})}{\partial a^{j}_{i,t+1}}a^{j}_{i,t+1}\right)+ \lambda_{i,t}
\end{eqnarray*} \normalsize
and the intratemporal labor supply condition, 
\begin{eqnarray*}
	l^{\eta}_{i,t} &=& \frac{W_{t}}{P_{t}}c_{i,t}^{-\gamma}
\end{eqnarray*}\normalsize where I have used that $\frac{\partial y_{i,t}(l_{i,t})}{\partial l_{i,t}}=\frac{W_{t}}{P_{t}}$, that is, the derivative of total individual income with respect to hours worked is the real wage. 
The parameters of interest are therefore $(\gamma,\beta,\eta)$.\footnote{ Note that the assumption of separability between consumption and labor can be relaxed without any bearing on the main point of the paper. With non-separability, aggregation will involve identifying an additional moment from micro data, that of the cross sectional covariation of consumption and hours worked, while it will introduce an additional parameter to be estimated, that of the degree of separability. Hence, separability does not buy as much in terms of the order condition for identification.}

\normalsize

For this application, I will focus on assets whose individual returns depend only on aggregate risk e.g. government bonds or savings accounts whose benchmark returns depend on the policy rate.\footnote{This is also consistent with a wider asset class whose returns are characterized by $r(a^{j}_{i,t+1},S_{t+1})=r(S_{t+1}) + \frac{\partial r_{j}(a^{j}_{i,t+1})}{\partial a^{j}_{i,t+1}}a^{j}_{i,t+1}$ where the second idiosyncratic component is unobserved. This is true for e.g. a credit line which involves higher interest rates at higher levels of unsecured borrowing and can be thought of as a softer borrowing constraint which has similar properties to $\lambda_{i,t}$.}
Aggregating as in \eqref{eq:aggeul} yields\small
\begin{eqnarray}
	C_{t}^{-\gamma}\Xi_{t}&=&\beta\mathbb{E}_{t}C_{t+1}^{-\gamma}{\Xi_{t+1}}{}(\underset{R_{t+1}}{\underbrace{1+r(S_{t+1})}}) + {\mu_{t}} \label{eq:distR}
\end{eqnarray}
 \normalsize
Furthermore, the intratemporal condition is:
\[L_{t}=W_{t}^{\frac{1}{\eta}}C_{t}^{-\frac{\gamma}{\eta}}\Xi^{lab}_{t}\] where\[\Xi^{lab}_{t}\equiv\int \Xi^{\frac{1}{\eta}}_{i,t}p(s_{i,t}|S_{t})d s_{i,t} \]
From the above expressions, it is clear that there are two main channels through which heterogeneity matters for aggregate fluctuations to first order. The first is fluctuations in the dispersion of the consumption share, which distorts both the intertemporal allocation of aggregate consumption as well as the intratemporal consumption - labor tradeoff even if all agents are "on their Euler equations". The aggregation residual, which is a cross sectional moment of the consumption share is essentially a sufficient statistic for the underlying moments of the wealth distribution as consumption is a nonlinear function of wealth.   Notice also that unless $\eta=1$, that is the Frisch elasticity is equal to one, the way heterogeneity matters for hours is different than consumption (the aggregation residuals are different), which implies that the dynamics of hours will also be affected by heterogeneity.
The second channel is the additional aggregate distortion $\mu_{t}$, which is positive whenever there is a subset of agents who are constrained.
  
Therefore, the aggregated moment inequalities that are relevant for identification are the following, where the second moment inequality arises from incorporating the aggregate intratemporal condition in the Euler equation: 
\begin{eqnarray}
	\mathbb{E}\left(1-\beta \left(\frac{C_{t+1}}{C_{t}}\right)^{-\gamma}\left(\frac{\Xi_{t+1}}{\Xi_{t}}\right)R_{t+1}\right) \phi(\mathbf{Y}^{t})&\geq& 0\\
	\mathbb{E}\left(1-\beta\left(\frac{\Xi_{t+1}}{\Xi_{t}}\right)\left(\frac{\Xi^{lab}_{t}}{\Xi^{lab}_{t+1}}\right)^{\eta}\left(\frac{L_{t+1}}{L_{t}}\right)^{\eta}\frac{W_{t}}{W_{t+1}}R_{t+1}\right) \phi(\mathbf{Y}^{t})&\geq& 0 \nonumber
\end{eqnarray} 
\subsection{Implementation}
The above moment inequalities are expressed as moment equalities using the component-wise positive valued nuisance vectors $\mu$ and $\kappa$ respectively,	where $R\in\{R_{g},R_{ib}\}$, are the ex-post real 3-month Treasury bond rate and the Interbank rate:
\small
	\begin{eqnarray}
	\mathbb{E}\left[\left(1-\beta \left(\frac{C_{t+1}}{C_{t}}\right)^{-\gamma}\left(\frac{\Xi_{t+1}}{\Xi_{t}}\right)R_{t+1}\right) \phi(\mathbf{Y}^{t})\right]-\mu_{1}&=& 0 \label{eq:noB}\\
	\mathbb{E}\left[\left(1-\beta\left(\frac{\Xi_{t+1}}{\Xi_{t}}\right)\left(\frac{\Xi^{lab}_{t}}{\Xi^{lab}_{t+1}}\right)^{\eta}\left(\frac{L_{t+1}}{L_{t}}\right)^{\eta}\frac{W_{t}}{W_{t+1}}R_{t+1}\right) \phi(\mathbf{Y}^{t})\right]-\mu_{2}&=& 0 \nonumber
	\end{eqnarray}  
	\begin{eqnarray}
		\mathbb{E}\left[\left(1-\beta \left(\frac{C_{t+1}}{C_{t}}\right)^{-\gamma}\left(\frac{\Xi_{t+1}}{\Xi_{t}}\right)R_{t+1}\right)B^{-1}_{t}\phi(\mathbf{Y}^{t})\right]-\kappa_{1}&=& 0  \label{eq:withB} \\
		\mathbb{E}\left[\left(1-\beta\left(\frac{\Xi_{t+1}}{\Xi_{t}}\right)\left(\frac{\Xi^{lab}_{t}}{\Xi^{lab}_{t+1}}\right)^{\eta}\left(\frac{L_{t+1}}{L_{t}}\right)^{\eta}\frac{W_{t}}{W_{t+1}}R_{t+1}\right)B^{-1}_{t} \phi(\mathbf{Y}^{t})\right]-\kappa_{2}&=& 0 \nonumber
	\end{eqnarray}  
	\normalsize
	
 To construct instruments $\phi(\mathbf{Y}^{t})$, I use lagged values of $(\Delta \mathcal{Y}_{t},\Delta I_{t})$, $I_{t}$ being aggregate investment and $\mathcal{Y}_{t}$ the real gross domestic product, and $(R_{g,t},R_{ib,t})$.
 
 In order to measure $B_{t}$, I combine information from the Survey of Household Finances (SHF) and the Business and Consumer Survey (BCS) of the European Commission \citet{BCS}.\footnote{SHF data can be downloaded from \href{https://www.bde.es/bde/en/areas/estadis/estadisticas-por/encuestas-hogar/relacionados/Encuesta_Financi/}{here} and BCS data can be downloaded from \href{https://ec.europa.eu/info/business-economy-euro/indicators-statistics/economic-databases/business-and-consumer-surveys_en}{here}.} 
I use the SHF to compute the exact triennial measure for $B_{t}$. The questionnaire asks households whether they have been denied a loan application during this period.\footnote{The survey question does not differentiate between business and private use.}  This question is consistent with a variety of reasons for which these households have been rejected, depending on individual characteristics, employment, guarantees, changes in the institution's credit policy and excessive debt. Moreover, I account for consumers that have not applied for loans because they expect to be rejected (also coined as "discouraged borrowers" in \citet{10.2307/2077836} ) and consumers that were granted a fraction of the requested loan. In order to impute $B_{t}$ in the periods in between, I employ the SHF and BCS data in a mixed frequency state space model. Appendix \hyperref[AppA]{A} (Section \ref{measurements} ) contains a detailed analysis of this approach. Both the SHF and BCS data as well as the extracted time series are plotted in Figure \ref{mixedLCC}.
 
 Finally, another variable that needs to be measured from micro data is the consumption share, which is important to construct the aggregation residuals $\Xi_{t}$ and  $\Xi^{lab}_{t}$. The measure is constructed using the Spanish Survey of Household Finances (SHF).\footnote{Please consult Section \ref{cons_sh} in Appendix \hyperref[AppA]{A} for a detailed explanation of how individual consumption is constructed.} In order to impute the values of aggregation residuals when they are latent, I use shape preserving interpolation. This is not likely to be a detrimental assumption for this kind of exercise as consumption is smoother than income, and its distribution is a slowly moving variable see e.g. \citet{bover}.\footnote{In particular, see  Figure 13 in that paper.} The performance of the method when employing this kind of interpolation is also checked in the Monte Carlo exercises and can be found in Appendix \hyperref[AppA]{A} (Section \ref{MCev}).

 \subsubsection{Details on Inference}
As in the simulation exercise, the parameter estimates for $(\theta,U)$ are defined as the minimizers of the continuously updated GMM criterion in \eqref{eq:cugmm}  where $\bar{m}(Y_{t},Y_{t+1},\theta)$ contains the sample counterpart of the above moment conditions with $\theta=(\gamma, \eta,\beta)$. The algorithm for inference used here is the same as the one employed in the simulation exercise.
%
The prior upper and lower bounds are $\gamma\in [0.01,20]$, $\frac{1}{\eta}\in [\frac{1}{15},\frac{1}{0.01}]$,  $\beta\in [0.8,0.9999]$. 

Below I present the empirical results in the case in which \eqref{eq:noB} and \eqref{eq:withB} are used separately, jointly, and when adding information using the average intratemporal condition as an additional moment.\footnote{\small
	\begin{eqnarray*}
		\mathbb{E}	\left(\left(\frac{C_{t+1}}{C_{t}}\right)^{-\gamma}- \left(\frac{\Xi^{lab}_{t}}{\Xi^{lab}_{t+1}}\right)^{\eta}\left(\frac{L_{t+1}}{L_{t}}\right)^{\eta}\frac{W_{t}}{W_{t+1}}\right) \phi(\mathbf{Y}^{t})&=&0
\end{eqnarray*} }	\normalsize Moreover, I choose not to use equity returns for estimation for two equally important reasons. First, in light of the results of \citet{ASCARI2021129}, stock market returns are likely to give imprecise estimates due to weak identification, while risk free returns less so. Second, I wish to revisit the predicted equity premium in order to provide empirical validation to the paper's approach; therefore, equity returns should not be used to fit the model.

\subsubsection{Results}The results that follow provide empirical support to Proposition  \ref{over_surv}. 
Table \ref{Lim_info_est_1} presents the $95\%$ confidence sets using $R_{ib}$ and $R_{g}$ as returns. 
Note that the confidence sets do not necessarily overlap, which is attributed to the different weighting that moments receive in finite sample, as in all cases the parameters are (conditional on $U$) overidentified.\footnote{This is similar to the more standard behavior of the GMM estimator in finite samples in the case of point identification. With more moments than parameters, since not all moments can be simultaneously zero, the weighting matrix has an impact on the final estimates. In this paper's setting, when using information on $B_{t}$, the estimate of the long run variance will be different than in the case in which $B_{t}$ is not utilized. Hence, the location of the estimates can be affected by the difference in the weighting scheme. This is also true when combining the two conditions  and when adding the intratemporal optimality condition in the set of estimating moments.}  

Comparing the sum of squared lengths of the confidence sets for all parameters in the case in which the extensive margin is utilized to the case in which it is not, we find that there is  a reduction that ranges from $71.2\%$ to $99.7\%$.

\begin{table}[h] 
	\small
	\begin{centering}
		\caption{Confidence Set for Preference Parameters}\bigskip
		\label{Lim_info_est_1}
		\resizebox{\columnwidth}{!}{
			\begin{tabular}{c c c c c c c }
				\hline 
				Parameter& \multicolumn{2}{c}{without $B_{t}$}\vline &  
				\multicolumn{2}{c}{with $B_{t}$}\vline& & \\  
				\hline &$q_{2.5\%}$ & $q_{97.5\%}$& $q_{2.5\%}$ & $q_{97.5\%}$&Instruments & Return
				\\ \hline \hline			
							$\gamma$  &  3.328  & 7.744 &  1.476  & 1.644 &   & \\                               
				$\frac{1}{\eta}$&  0.140 & 0.285 & 0.254 & 0.271 & $\Delta \mathcal{Y}_{t-1},\Delta I_{t-1}$, $R_{ib,t}$, & {$R_{ib,t}, R_{g,t}$} \\
				$\beta$ & 0.969  & 0.999&  0.998  &   0.999  & &\\  \hline		
					\multicolumn{7}{c}{Combination}\\  \hline
				$\gamma$& 5.614	   &10.362& 3.610 &5.410 &  &  \\                               
				$\frac{1}{\eta}$   & 0.146 & 0.522 & 0.067 & 0.075&  $\Delta \mathcal{Y}_{t-1},\Delta I_{t-1}$, $R_{ib,t}$, & {$R_{ib,t}, R_{g,t}$}  \\
				$\beta$  & 0.966& 0.999 &  0.983 &   0.999  &    & \\\hline\hline  
					\multicolumn{7}{c}{Adding Intratemporal Condition}\\  \hline
					$\gamma$& 3.864	   &4.573& 1.845 &2.111 &  &  \\                               
				$\frac{1}{\eta}$   & 0.787 & 1.056& 0.095 & 0.098&  $\Delta \mathcal{Y}_{t-1},\Delta I_{t-1}$, $R_{ib,t}$, & {$R_{ib,t}, R_{g,t}$}  \\
				$\beta$  & 0.992& 0.999 &  0.996 &   0.999  &    & \\\hline\hline  
		\end{tabular}}	
		\par\end{centering} 
\end{table}
\normalsize

In all cases, had we neglected the extensive margin of constrained consumers, higher values of risk aversion would be admissible. This is also consistent with the results of Section 3 and the analytical results in Appendix \hyperref[AppB]{B}. The time varying proportion of constrained agents provides additional information. 
Looking at the smallest interval compatible with the different estimates, risk aversion ranges from 1.48 to 1.64, which is close to estimates obtained in the microeconometric literature and much smaller than the typical estimates obtained using frictionless representative agent models.  
Allowing for positive distortions in the Euler equation for bonds implies that smooth consumption growth is consistent with the more volatile returns without requiring implausibly large estimates of risk aversion. This also implies that the  elasticity of intertemporal substitution is not too low, addressing the risk free rate puzzle as well, see e.g \cite{WEIL1989401}.

Turning attention to the Frisch elasticity of labor supply, the confidence sets for the Frisch elasticity are also refined and much more sharp after accounting for the extensive margin, providing further support to Proposition  \ref{over_surv}. The estimated elasticities are much less than 1 in general, close to the values supported by microeconomic estimates. In fact, as \citet{10.1257/aer.101.3.471} suggest,  models that require an elasticity above 1 are inconsistent with micro evidence. While in the first two sets of results estimates were already much less than 1 even before accounting for the extensive margin, this is not true in the third case that employs the intratemporal condition itself as an additional identifying restriction. Employing the extensive margin substantially lowers the estimated elasticity.  

As we did with the consumption Euler equation in \eqref{eq:intui}, we can correspondingly build intuition for the identification of the Frisch elasticity by inspecting the log-linear intratemporal condition.
Denoting the marginal utility of consumption by $MU_{t}:=C_{t}^{-\gamma}$, the Frisch elasticity is potentially  identifiable using the elasticity ($\epsilon$) of different macroeconomic aggregates with respect to  the wage rate i.e.  hours ($L$), the aggregation residual ($\Xi^{lab}$), the undistorted marginal utility of wealth ($MU-\mu$) and the extensive and intensive margins of the distortion, $B$ and $\kappa$ respectively:
\begin{eqnarray*} \frac{1}{\eta}=\frac{\frac{\partial ln L_{t}}{\partial ln W^{}_{t}}-\frac{\partial ln \Xi^{lab}_{t}(\eta)}{\partial ln W^{}_{t}}}{1+\frac{\partial ln MU_{t}}{\partial ln W^{}_{t}}}&=&\frac{\epsilon_{L,W}-\epsilon_{\Xi^{lab}(\eta,\gamma),W^{}}}{1+\epsilon_{MU-\mu,W^{}} +(\epsilon_{B,W^{}}+\epsilon_{\kappa,W^{}})\frac{\mu}{MU-\mu}}
\end{eqnarray*}
This is an implicit equation in $\eta$, which  identifies $\eta$ by adjusting  the raw elasticity $\epsilon_{L,W}$ for variation that is due to wealth effects in order to isolate the substitution effect. The admissible values of $\eta$ will depend non-linearly in all of these elasticities and on the admissible values of $\gamma$. The more responsive are dispersion $(\Xi^{lab})$ and the two margins of the distortion $(B,\kappa)$ to changes in the wage rate, the greater the divergence between the wage elasticity of aggregate hours from $\frac{\epsilon_{L,W}}{1+\epsilon_{MU,W}}$, which is the Frisch elasticity in the complete markets case. Being able to measure $B$ implies that the admissible estimates of $\frac{1}{\eta}$ will only vary with the admissible domain of the wage elasticity of the intensive margin of adjustment, $\epsilon_{\kappa,W}$. 
 
Consistent with the results of \citet{DOMEIJ2006242}, the adjustment to individual behavior in the presence of borrowing constraints can imply a higher elasticity of labor supply as $\epsilon_{\kappa,W}$ is expected to be negative. Moreover, the extensive margin $B_{t}$ affects aggregate labor supply in the same way. This is also in line with the fact that the lower bound on the Frisch elasticity is refined upwards in  the first set of results when $B_{t}$ was employed. 
Nevertheless, this is not true in the other cases. 
When macro aggregates are utilized for identification, higher estimates of the Frisch elasticity would be the only possible refinement when $\epsilon_{\Xi^{lab}(\eta,\gamma),W^{t}}$ is negligible. When the latter is not negligible, then identification of $\eta$ depends in a more complicated way on the rest of the parameters e.g. risk aversion and the aggregation residual $\Xi_{t}$.

 Finally, regarding the discount factor,  there are no major information gains, and estimates range from 0.983 to 0.999. Consistent with the simulation experiment, accounting for the extensive margin refines the estimates of $\beta$ upwards.  


\subsection{Does improved identification matter? }
Having shown that accounting for distortions in the Euler equation does matter for preference identification, I next turn to analyzing in more detail how this is likely to affect inference on asset prices. Since we have not utilized equity returns to estimate preferences, it is instructive to see whether the estimated preferences can also explain the observed equity premium in the Spanish Stock Market. 
Using the IBEX35 Total returns index, which includes the 35 most liquid stocks in the Madrid Stock Exchange and is available from 2001Q1 onwards, I compute the mean risk premium over government bonds in the period 2001Q1-2018Q and the corresponding prediction of the risk premium as implied by the model, accounting for the two distortions as follows:
\begin{eqnarray}{ln\left(\frac{\mathbb{E}R^{eq}}{\mathbb{E}R^{g}}\right)}\approx - \frac{Cov(\mathcal{M}(\theta),R^{eq})+\mu^{eq}(\theta)-\mu^{g}(\theta)}{1-\mu^{g}(\theta)}  \label{riskp}
\end{eqnarray}
where $\mathcal{M}(\theta):=\beta \left(\frac{C_{t+1}}{C_{t}}\right)^{-\gamma}\left(\frac{\Xi_{t+1}}{\Xi_{t}}\right)$ is the heterogeneity adjusted stochastic discount factor, evaluated at all $\theta\in\Theta_{95\%}$, $R^{eq}$ is the real equity return index and $(\mu^{eq},\mu^{g})$ are the average distortions in the Equity and Bond Euler equations respectively, also evaluated at $\theta\in\Theta_{95\%}$.\footnote{Notice that these distortions can be easily computed by plugging in $\theta$ and computing the sample mean of the corresponding Euler Equations e.g. \[\hat{\mu}^{g}(\theta)=\frac{1}{T}\sum_{t=2..T}\left( C_{t}^{-\gamma}-\beta C_{t+1}^{-\gamma}\frac{\Xi_{t+1}}{\Xi_{t}}R_{g,t+1}\right)\]}

In Figure \ref{eqp} we observe that the data based prediction for the average risk premium lies very close to the set valued predictions based on the different estimates. Moreover, shutting down frictions and heterogeneity delivers a counterfactual prediction for a weakly negative risk premium. What this implies is that had we not accounted for incomplete markets and the possibility of trading constraints across assets, in order to match the observed premium, parameter estimates would have to be severely distorted. 

This consists of indirect evidence on the relevance of incomplete markets and financial constraints for the equity premium puzzle \citep{10.2307/1832056,MEHRA1985145,10.2307/1815722} as these frictions are both qualitatively and quantitatively important.\footnote{Several mechanisms have been considered as potential resolutions to both the risk premium and the risk free puzzle e.g. uninsurable income risk and borrowing constraints, limited asset market participation and transaction costs  (\citet{10.1162/003355302753399508,10.1086/340776,10.2307/2138860, AIYAGARI1991311}), as well as non-time separable preferences \citep{doi:10.1111/j.1540-6261.2004.00670.x, RePEc:eee:moneco:v:24:y:1989:i:3:p:401-421}, lower tail risk \citet{NBERw11310}, misperceptions in returns \citep{10.1257/000282803321947407} and consumption habits \citep{10.2307/2937698}. Yet, results are sensitive to the particular model setups, parameterizations and wider implications (see e.g. \citet{4137daa05d1449ddbf4171960bde2680,10.1257/jep.23.1.193} for a review.).} 
The moment inequality approach allows for the possibility that the Euler equation does not hold with equality, with distortions having a certain direction in terms of sign. This enables the model to accommodate for such distortions without forcing  preference parameters to take implausible values. In addition, since preferences are (set) identified, then we can readily compute these distortions (i.e. $\mu^{eq}$ and $\mu^{g}$) and hence properly correct for their presence in the computation of the risk premium in \eqref{riskp}. 

\begin{figure}[h]
	\begin{centering}			
		\includegraphics[scale=0.35]{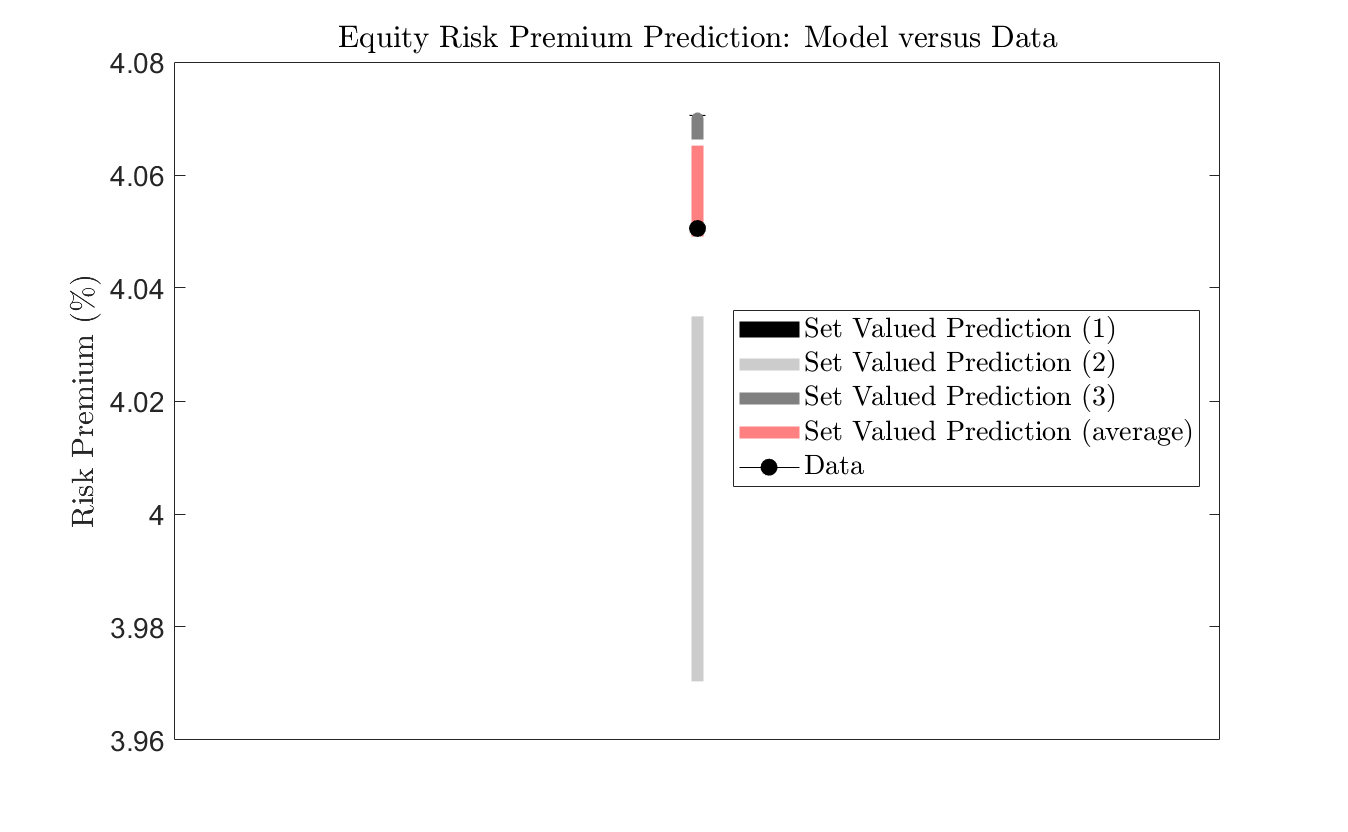}
		\includegraphics[scale=0.52]{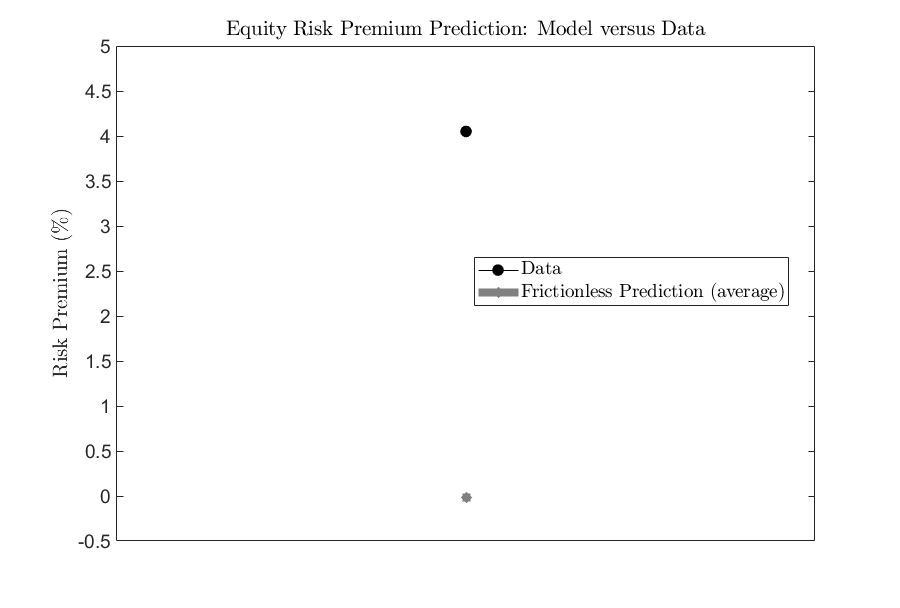}
		\par
	\end{centering}  \vspace{-0.2 in}	\protect\caption{The top figure compares the predicted equity risk premium based on the covariance between the heterogeneity consistent stochastic discount factor and the IBEX35 total returns index, adjusted by the Euler equation distortions. The different estimates correspond to the three cases in Table \ref{Lim_info_est_1} and their average. The lower figure computes the counterfactual prediction for $\mu^{eq}=\mu^{g}=0$ and $\Xi_{t}=1,\forall t$.}
	\label{eqp}  
\end{figure}  

Furthermore, we can utilize the estimates of the distortions to Euler equations to judge which trading restrictions are more severe. This exercise, which is similar in spirit to business cycle accounting (\citet{ECTA:ECTA768}), can shed light on the size of the distortions that fully specified structural models should be able to replicate for each asset class.\footnote{Correspondingly, these estimates may also be used as target statistics to validate structural models. In \citet{tryphonides2017set}, I develop a formal test.}
In contrast to typical business cycle accounting exercise where parameters are externally calibrated, the moment inequality framework has allowed us to estimate the structural parameters first while being consistent with the class of frictions under consideration. Thus, the estimated wedges will reflect both the in-sample information, as well as the model uncertainty that is already encoded by the identified set.

Figure \ref{bca} plots the estimated distortions to the two Euler equations used in estimation, as well as the Equity Euler equation, where estimates are expressed as a share of marginal utility.
\begin{figure}[]
	\begin{centering}			
		\includegraphics[scale=0.52]{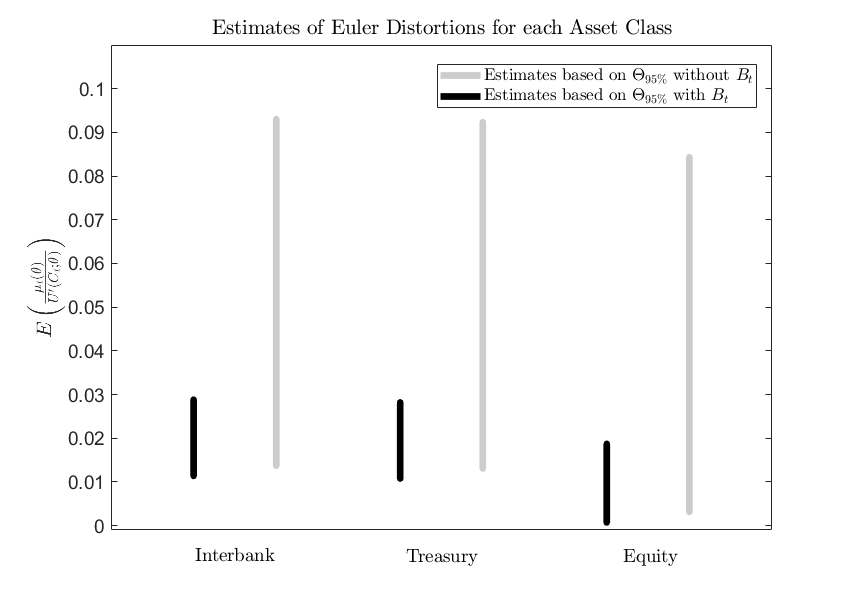}
		\par
	\end{centering}  	\protect\caption{Black lines are the distortions $(\mu^{ib},\mu^{g},\mu^{eq})$ estimated using $\Theta_{95\%}$ from Table 1, Case 1 (with $B_{t}$) while gray lines are the estimates using  $\Theta_{95\%}$ from Table 1, Case 1 (without $B_{t}$). }
	\label{bca}  
\end{figure} 
 Distortions in the bond market range from $1.1\%$ to $2.7\%$ of marginal utility, where distortions in equity trading  range from $0.1\%$ to $2.0\%$. 
In order to provide further interpretation to these distortions, the bond and equity pricing equations may be rewritten as follows: \small
\begin{eqnarray}
	\mathbb{E}_{t}\left(\mathcal{M}_{t+1}R^{g}_{t+1}-1\right)&=& - \frac{\mu^{g}_{t}}{U'(C_{t};\gamma,\eta)}\\
	\mathbb{E}_{t}\left(\mathcal{M}_{t+1}\left(R^{eq}_{t+1}-R^{g}_{t+1}\right)\right)&=&  \frac{\mu^{g}_{t}-\mu^{eq}_{t}}{U'(C_{t};\gamma,\eta)}
\end{eqnarray} \normalsize
The first expression shows the negative distortion to the risk free rate return due to borrowing constraints, that is, there is "excessive demand" for bonds which puts downward pressure on the return. This is consistent with the existing literature that seeks to rationalize the risk free rate puzzle through incomplete markets. The second distortion is the weighted risk premium, where the weight is the stochastic discount factor. The average distortion $\mu^{eq}$ could be either positive or negative. This can be rationalized by the fact that with measured total returns, which include dividend payments but no transaction costs, the distortion to the equity Euler equation will inevitably soak up the aggregate unobserved transaction costs. 

Utilizing a generalized version of \eqref{eq:genzeldes} in the case of equities $(a^{eq})$ that involve transaction costs $\psi(a^{eq}_{t+1},a^{eq}_{t})$, assuming that they are constant across agents, yields an aggregate distortion to the Euler equation as follows:\footnote{Note that $\psi_{1}$ and $\psi_{2}$ are the marginal transaction costs with respect to the first and second argument. See Appendix \hyperref[AppA]{A}  for the derivation of the generalized version of \eqref{eq:genzeldes}.}\small
\begin{eqnarray}
	\frac{\mu^{eq}_{t}}{U'(C_{t};\gamma,\eta)}&\equiv&\frac{\mu_{t}}{U'(C_{t};\gamma,\eta)}+ \mathbb{E}_{t}\mathcal{M}_{t+1}\left(  R^{eq}_{t+1} \left(  \frac{1}{1+\psi_{1,t+1}}-1\right)-\frac{\psi_{2,t+1}}{1+\psi_{1,t+1}} \right)  
\end{eqnarray}   
\normalsize Depending on the signs and magnitudes of the unobserved marginal transaction costs, the contribution to $\frac{\mu^{eq}_{t}}{U'(C_{t};\gamma,\eta)}$ can be negative, while the trading constraint e.g. the restriction in short selling has a positive contribution ($\mu_{t}$). The overall effect is therefore ambiguous. 
A negative value for  $\frac{\mu^{eq}_{t}}{U'(C_{t};\gamma,\eta)}$ provides indirect evidence for stronger net transaction costs, and it positively contributes to the risk premium, while a positive value is evidence of stronger trading constraints.\footnote{As an example of transaction costs, consider the case in which $\psi$ is symmetric e.g. quadratic in $a^{eq}_{i,t+1}-a^{eq}_{i,t}$. This implies that ${\psi_{2,t+1}}=-\psi_{1,t+1}$ while $\psi_{1,t+1}$ can be large enough such that $ R^{eq}_{t+1} \left(  \frac{1}{1+\psi_{1,t+1}}-1\right)-\frac{\psi_{2,t+1}}{1+\psi_{1,t+1}}$ becomes negative, making shares a costly asset to hold and thus increasing the required premium. Conversely, trading constraints such as short selling restrictions increase $\mu^{eq}_{t}$, exerting downward pressure on the premium. See \citet{AIYAGARI1991311,10.2307/2138860} for quantitative results on the effect of transaction costs on the equity premium as well as  \citet{LAGOS2010913} for a liquidity interpretation of the puzzle.}
The identified distortion is positive, hence providing evidence of stronger trading constraints.

Moreover, it is interesting to see whether the results change when we employ preference parameter estimates that do not incorporate information on the extensive margin of constrained consumers. Results in grey colour indicate that had we not accounted for the borrowing constrained, parameter estimates (higher risk aversion in particular) would overpredict the level of distortions. Information on the extensive margin refines the size of the identified distortions.

  \section{Conclusion}
  
 
Financial constraints at the household level generate distortions to the standard representative agent Euler equation that have important implications for preference identification. This paper proposed moment inequality restrictions that account for these features and has shown theoretically and empirically that incorporating information on the proportion of financially constrained households can substantially improve estimates and subsequent inference. While the aggregate restrictions can cater to some limited types of ex-ante heterogeneity, we can extend the moment inequality approach to accommodate more general types using panel data.



\bibliographystyle{econometrica}
\bibliography{biblio_thesis}

\renewcommand{\theequation}{A.\arabic{equation}}

\setcounter{equation}{0}
 \section{Appendix A} \label{AppA}
 
 This Appendix contains the following material:  (1)  Model extensions that account for multiple assets and transaction costs and ex-ante heterogeneity, (2) Lemmas  \ref{nr} and  \ref{over} and their proofs as well as proofs of  results  \eqref{eq:no_surv_mom} and \eqref{eq:surv_mom} and Proposition \ref{over_surv} in the paper, (3)  An example of a wealth distribution with a point mass at the borrowing constraint, (4) Additional Simulation Evidence,  (5) Details on imputing $B_{t}$ when not directly observed and the mixed frequency model, (6)  Details on Data.

 \subsection{Household Optimization with Multiple Assets and Transaction Costs}
The household chooses to store wealth in $J$ assets  $\left\{a^{j}_{i,t}\right\}_{j=1..J}$ which earn gross returns equal to $R^{j}_{i,t}=1+r^{j}(a^{j}_{i,t},S_{t})$,  while $\psi(a^{j}_{i,t+1},a^{j}_{i,t})$ are the transaction costs for the $j^{th}$ asset. 
 Furthermore, there is a general constraint technology $G_{i}(.)$ that restricts trades between these assets which is increasing in wealth. 
 
 Let $\beta\in (0,1)$ be the discount factor and $U(c_{i,t},l_{i,t};\omega,S_{t})$ the instantaneous utility function, where $\omega$ signifies parameters that are related to preferences  and $S_{t}$ an aggregate state vector. The  household problem is as follows:
 \begin{eqnarray*}
 	&&\max_{\{c_{i,t},l_{i,t},\left\{a^{j}_{i,t+1}\right\}_{j=1..J}\}_{t=0}^{\infty }}\mathbb{E}_{0}\sum_{t=0}^{\infty }\beta
 	^{t}U(c_{i,t},l_{i,t};\omega,S_{t})\\
 	s.t. &  c_{i,t} -  y_{i,t}(l_{i,t})   &= \sum \limits^{J}_{j=1}\left(1+r_{j}(a^{j}_{i,t},S_{t})\right)a^{j}_{i,t} - \displaystyle \sum \limits^{J}_{j=1} a^{j}_{i,t+1} - \sum \limits^{J}_{j=1}\psi(a^{j}_{i,t+1},a^{j}_{i,t}) \\ 
 	& c_{i,t}>0&\text{, } G_{i}\left(\left\{a^{j}_{i,t},a^{j}_{i,t+1}\right\}_{j=1}^{J}\right)\geq 0   \label{eq:model}
 \end{eqnarray*} 
 Denoting the marginal utility of consumption by $U'_{1}(c_{i,t};l_{i,t},\omega,S_{t})$, 
 and marginal transaction costs for $(a^{j}_{i,t},a^{j}_{i,t+1})$ by  $\psi^{j}_{1}$ and $\psi^{j}_{2}$ respectively, the Euler equation for any asset $j$ is distorted by the non-negative Lagrange multipliers on the
 occasionally binding constraint $(\mu_{i,t})$ which are positive when the constraint is binding:
 \footnotesize
 \begin{eqnarray}
 	&& U'_{1}(c_{i,t},l_{i,t};\omega,S_{t})(1+\psi^{j}_{1,i,t+1})\nonumber\\
 	& = & \beta \mathbb{E}_{t}\left(1+r_{j}(a^{j}_{i,t+1},S_{t+1})+\frac{\partial r_{j}(a^{j}_{i,t+1},S_{t+1})}{\partial a^{j}_{i,t+1}}a^{j}_{i,t+1}-\psi^{j}_{2,i,t+1}\right)U'_{1}(c_{i,t+1},l_{i,t+1};\omega,S_{t+1}) \nonumber\\  && +\mu_{i,t}\frac{\partial G_{i}\left(\left\{a^{j}_{i,t},a^{j}_{i,t+1}\right\}_{j=1}^{J}\right)}{\partial a^{j}_{i,t+1}}+\beta\mathbb{E}_{t}\mu_{i,t+1}\frac{\partial G_{i}\left(\left\{a^{j}_{i,t+1},a^{j}_{i,t+2}\right\}_{j=1}^{J}\right)}{\partial a^{j}_{i,t+1}}  \nonumber 
 \end{eqnarray}\normalsize
 which can be simplified to  \footnotesize
 \begin{eqnarray}
  	&&  \quad \quad U'_{1}(c_{i,t},l_{i,t};\omega,S_{t})(1+\psi^{j}_{1,i,t+1})  \\
 	& = &
 	\beta \mathbb{E}_{t}\left(1+r_{j}(a^{j}_{i,t+1}, S_{t+1})+\frac{\partial r_{j}(a^{j}_{i,t+1},S_{t+1})}{\partial a^{j}_{i,t+1}}a^{j}_{i,t+1}-\psi^{j}_{2,i,t+1}\right)U'_{1}(c_{i,t+1},l_{i,t+1};\omega,S_{t+1})\nonumber + \lambda^{j}_{i,t} \nonumber
 \end{eqnarray}
 \normalsize where $\lambda^{j}_{i,t}$ summarizes the distortions for each asset $j$, and is positive as $\mu_{i,t}$ is positive and $G_{i}$ is increasing in wealth. The effective rate or return for asset $j$ is defined as,\small \[r^{e,j}_{i,t+1}:=\frac{r_{j}(a^{j}_{i,t+1}, S_{t+1})+\frac{\partial r_{j}(a^{j}_{i,t+1},S_{t+1})}{\partial a^{j}_{i,t+1}}a^{j}_{i,t+1}-\psi^{j}_{1,i,t+1}-\psi^{j}_{2,i,t+1}}{1+\psi^{j}_{1,i,t+1}}\]\normalsize
 Hence, the individual Euler equation for asset $j$  can be more compactly written as \small
 \begin{eqnarray}
 	U'_{1}(c_{i,t},l_{i,t};\omega,S_{t})& = &
 	\beta\mathbb{E}_{t}U'_{1}(c_{i,t+1},l_{i,t+1};\omega,S_{t+1})\left(1+r^{e,j}_{i,t+1}\right) + {\hat{\lambda}}^{j}_{i,t} 
 \end{eqnarray}
 \normalsize
 \subsection{Ex Ante Heterogeneity}  \vspace{-0.2 in}
 \subsubsection{Heterogeneity in Preferences}
 The proposed approach may be able to capture some but not all forms of ex-ante heterogeneity. In some cases, the utility specification may permit accommodating for preference heterogeneity using aggregate information. \citet{10.1257/aer.104.7.2075} use a similar utility function to the one used in this paper's application but incorporate heterogeneity in dis-utility from work as follows:  \[u(c_{i,t},l_{i,t},\gamma,\eta):= \frac{c_{i,t}^{1-\gamma }-1}{1-\gamma }- e^{\xi_{i}}\frac{l_{i,t}^{1+\eta}}{1+\eta}   \] According to their results, allowing for heterogeneity in this parameter is important for matching the cross sectional joint distribution over wages, hours and consumption. The corresponding intratemporal condition is $
 	e^{\xi_{i}}l^{\eta}_{i,t} = \frac{W_{t}}{P_{t}}c_{i,t}^{-\gamma}
$.
 The distribution of individual states now includes that of preferences. Aggregating,
 \begin{eqnarray*}
 	\int e^{\frac{\xi_{i,t}}{\eta}}l_{i,t} p(s_{i,t},\xi_{i}|S_{t})d (s_{i,t},\xi_{i})   &=& \left(\frac{W_{t}}{P_{t}}\right)^{\frac{1}{{\eta}}}C_{t}^{-\frac{\gamma}{\eta}} \Xi^{lab}_{t} 
 \end{eqnarray*} Using the fact that hours should co-vary negatively with the preference parameter i.e. a higher $\xi_{i,t}$ should be associated with lower labor supply, we obtain another  inequality as follows:
 \begin{eqnarray*}
 	\int e^{\frac{\xi_{i,t}}{\eta}}l_{i,t} p(s_{i,t},\xi_{i}|S_{t})d (s_{i,t},\xi_{i})   &\leq& \int e^{\frac{\xi_{i,t}}{\eta}}p(\xi_{i})d\xi_{i}\int l_{i,t}p(s_{i,t},\xi_{i}|S_{t})d (s_{i,t},\xi_{i})  
 \end{eqnarray*} and hence \vspace{-0.1 in} \begin{eqnarray*}
 	\left(\frac{W_{t}}{P_{t}}\right)^{\frac{1}{{\eta}}}C_{t}^{-\frac{\gamma}{\eta}} \Xi^{lab}_{t}  &\geq& L_{t}\int e^{\frac{\xi_{i,t}}{\eta}}p(\xi_{i})d\xi_{i}
 \end{eqnarray*}
 If we are willing to impose a flexible parametric assumption of the distribution of $\xi_{i}$, $p(\xi_{i};\varphi)$ then the moment condition may be constructed analytically or by simulation, and $\varphi$ can be (set) identified together with the rest of the parameters. Otherwise, the moment $\int e^{\frac{\xi_{i,t}}{\eta}}p(\xi_{i};\varphi)d\xi_{i}$ may be estimated non parametrically, as it is simply a constant.
 Accommodating for preference heterogeneity is therefore possible but context specific.\vspace{-0.1 in}
 \subsubsection{Heterogeneity in Expectations}  
 Inspecting the Euler equation in \eqref{eq:aggeul}, expectations about individual and aggregate states are treated as rational and homogeneous across agents. I next show that once we define expectations about e.g the aggregate state as belonging to some well defined family of deviations from rational expectations based on some observable, then it is in principle possible to accommodate for some form of expectation heterogeneity as well.  
 Let us define the subjective model so that the subjective probability $\hat{p}(S_{t+1}|S_{t},s_{i,t})$ has a finite Kullback-Leibler distance from the true probability model ${p}(S_{t+1}|S_{t})$, that is, $\int ln\left(\frac{\hat{p}(S_{t+1}|S_{t},s_{i,t})}{{p}(S_{t+1}|S_{t})}\right){p}(S_{t+1}|S_{t})dS_{t+1}<\infty$.   
  For simplicity, suppose that one of the aggregate states is future output, and each agent assigns some weight on individual income/unemployment expectations $(y^{e}_{i,t+1})$ in order to form expectations about the aggregate state. An example of such a subjective model is as follows:  \[\frac{\hat{p}(S_{t+1}|S_{t},s_{i,t})}{{p}(S_{t+1}|S_{t})}=e^{\varphi y^{e}_{i,t+1}},\quad y^{e}_{i,t+1}\in s_{i,t}\]  
 Hence, the aggregated expectations in the Euler equation will be equal to\footnote{This is derived as follows: \begin{eqnarray*}
 		&& \int\left[\int(1 + r_{t+1})c^{-\gamma}_{i,t+1}p(s_{i,t+1},S_{t+1}|s_{i,t},S_{t})d(s_{i,t+1},S_{t+1})\right]\hat{p}(s_{i,t}|S_{t})d(s_{i,t})\\
 		&=& \int\left[\int(1 + r_{t+1})e^{\varphi y^{e}_{i,t+1}}c^{-\gamma}_{i,t+1}p(s_{i,t+1},S_{t+1},s_{i,t}|S_{t})d(s_{i,t+1},S_{t+1},s_{i,t})\right]\\
 		&=& \int\left[\int(1+ r_{t+1})c^{-\gamma}_{i,t+1}e^{\varphi y^{e}_{i,t+1}}p(s_{i,t+1},s_{i,t}|S_{t+1},S_{t})d(s_{i,t+1},s_{i,t})\right]p(S_{t+1}|S_{t})d(S_{t+1})\\
 		&=& \mathbb{E}_{t}C^{-\gamma}_{t+1}(1 + r_{t+1})\left[\int \Xi_{i,t+1}e^{\varphi y^{e}_{i,t+1}}p(s_{i,t+1}|S_{t+1})d(s_{i,t+1})\right]
 \end{eqnarray*}}
 \begin{eqnarray*}
 	\mathbb{E}_{t}C^{-\gamma}_{t+1}(1 + r_{t+1})\left[\int \Xi_{i,t+1}e^{\varphi y^{e}_{i,t+1}}p(s_{i,t+1}|S_{t+1})d(s_{i,t+1})\right]
 \end{eqnarray*}
 
 The quantity $\int \Xi_{i,t+1}e^{\varphi y^{e}_{i,t+1}}p(s_{i,t+1}|S_{t+1})d(s_{i,t+1})$ can be estimated using data on consumption and income expectations from the repeated cross sections e.g. $\frac{1}{N}\sum_{i=1..N}\Xi_{i,t+1}(\gamma)e^{\varphi y^{e}_{i,t+1}}$ for each candidate parameter value. Expanding the set of predictors on which the subjective model is based is straightforward.


\subsection{Proofs and auxiliary Lemmas}\label{proofs} 
\begin{proof} of {Results}  \eqref{eq:no_surv_mom} and \eqref{eq:surv_mom}
	\begin{enumerate}[leftmargin=*]
		\item Starting from the conditional moment function which is separable in $\lambda^{o}_{t}$,   \[\mathbb{E}_{t}(g\left(Y_{t},Y_{t+1},\theta\right) -\lambda^{o}_{t})=\lambda^{un}_{t} \vspace{-0.1 in}\] Multiplying both sides by instrument $X_{t}>0$ maintains the same sign. Taking unconditional expectations we conclude. 
		\item Since $\mu_{t} = \kappa_{t}B_{t}$, dividing both sides by $B_{t}$ and taking unconditional expectations using $X_{t}$ we conclude. 		
	\end{enumerate}\vspace{-0.1 in}
\end{proof}  \vspace{-0.2 in}
 \subsubsection{\textbf{Identification Results}}\label{IdRes}
 I first define two notions that will be helpful for the theorems that follow: 
 \begin{defn} Denote moment function $m(Y_{t},Y_{t+1},\theta)$ as \underline{correctly specified} if there exists $\theta_{0}\in \Theta$ such that if $U=0$, $\mathbb{E}m(Y_{t},Y_{t+1},\theta_{0})=0$.\label{correct}
 \end{defn}\vspace{-0.3 in}
 \begin{defn} For some $U\in \mathbb{R}^{\mathfrak{m}}_{+}$, $\theta_{1}(U)$ is \underline{conditionally identified} if there does not exist any other $\theta\in \Theta$, $\theta_{2}(U)$ such that $\mathbb{E}m(Y_{t},Y_{t+1},\theta_{1}(U))=\mathbb{E}m(Y_{t},Y_{t+1},\theta_{2}(U))$. \label{unique} 
 \end{defn}\vspace{-0.2 in} The first ensures that the identified set is not empty, while the second guarantees that more information on the wedge $U$ leads to a smaller identified set for $\theta$. Given these definitions, Lemma \ref{nr} establishes that fixing a value for the wedge $U$, there exists a unique value for $\theta$, defined as $\theta(U)$, that satisfies \eqref{mom_comp_gmm}. Any additional information about $U$ is likely to  make the identified set smaller. 
 \begin{lem} {\textbf{$\Theta_{I}$ in the benchmark case.}}\\
 	Under correct specification of $m(.)$, conditional identification and full rank Jacobian matrix for $\mathbb{E}m(Y_{t},Y_{t+1},\theta)$, 
	\begin{center}$\exists$ bijective mapping $\mathcal{G}:\Theta_{I}=\mathcal{G}^{-1}(\mathcal{U})\cap \Theta$	\end{center} \label{nr}  
 \end{lem} 
 \begin{proof} of Lemma \ref{nr}\\
	Fix some $\underset{n_{\theta}\times 1}{U}\in  (\underline{U},\bar{U}]\equiv \mathcal{U}$. Assuming conditional identification as in Definition 2, there is a unique value of $\theta^{*}$ such that $\mathbb{E}m(Y_{t},Y_{t+1},\theta^{*}) \equiv\mathcal{G}(\theta^{*}) = \underset{n_{\theta}\times 1}{U}$. By the inverse function theorem (see e.g. p.266 in \citet{luenberger1997optimization}), a full rank Jacobian of the moment conditions implies that there is some open set ${\bar{\Theta}}$ that contains $\theta^{*}$ and an open set  $\mathcal{M}$ that contains $\mathcal{G}(\theta^{*})$ such that $\mathcal{G}:\mathcal{M}\to {\bar{\Theta}}$ has a continuous inverse $\mathcal{G}^{-1}:\bar{\Theta}\to \mathcal{M}$. Thus, by definition of the identified set, $\Theta_{I} =\{\theta\in\Theta: \mathcal{G}(\theta)=U,\forall U\in\mathcal{U}\}=\{\theta\in\Theta: \theta=\mathcal{G}^{-1}(U), \forall U\in\mathcal{U}\}=\bigcup_{U\in\mathcal{U}}\left(\mathcal{G}^{-1}(U)\cap \Theta\right)=\mathcal{G}^{-1}(\bigcup_{U\in\mathcal{U}})\cap\Theta={G}^{-1}(\mathcal{U})\cap\Theta$.\footnote{Note that this result can be extended to the case when $\mathfrak{m}> n_{\theta}$: Denoting the first $n_{\theta}$ conditional moments by $q_{1,t}(\theta)$ and the rest by $q_{2,t}(\theta)$ yields that: 
		\begin{eqnarray*}
			\underset{n_{\theta}\times1}{q_{1,t}(\theta)} & = & \mathbb{\mathcal{V}}_{1,t}\in[\underbar{\ensuremath{\mathcal{V}_{1}}},\bar{\mathcal{V}}_{1}]\\
			\underset{\mathfrak{m}-n_{\theta}\times1}{q_{2,t}(\theta)} & = & \mathbb{\mathcal{V}}_{2,t}\in[\underbar{\ensuremath{\mathcal{V}_{2}}},\bar{\mathcal{V}}_{2}]
		\end{eqnarray*}
		which imply the following unconditional moment equalities : 
		\begin{eqnarray*}
			\mathbb{E}((q_{1,t}(\theta)-\mathbb{\mathcal{V}}_{1,t})\otimes \phi_{t}) & = & 0\\
			\mathbb{E}((q_{2,t}(\theta)-\mathbb{\mathcal{V}}_{2,t})\otimes \phi_{t}) & = & 0
		\end{eqnarray*}
				Dropping momentarily the dependence on $\theta$ and on $t$, since $\mathfrak{m}>n_{\theta}$,
		and given a real
		valued matrix $W$  diagonal in $(W_{1 },W_{2
		})$, Lemma \ref{nr} follows through for the \textit{effective} first order conditions:  	\begin{eqnarray*}
			\underset{n_{\theta\times1}}{m}:=	\mathbb{E}J^{T}(W_{1}^{\frac{1}{2}T}q_{1}+W_{2}^{\frac{1}{2}T}q_{2}-\mathcal{V})\otimes \phi & = & 0
		\end{eqnarray*} 
		where the Jacobian $\underset{\mathfrak{m}\times n_{\theta}}{J}$ has full rank.}
  \end{proof} 
 
 \subsubsection{The use of additional information} Suppose now that, as in \eqref{eq:decomp}, each moment function in $m(.)$ can be rewritten as the sum of two components denoted by $m_{1}(\theta )$ and  $m_{2}(\theta )$ whose conditional expectation is positive. For brevity, I have dropped the dependence on the data.  Furthermore, denote by $Q_{1}$ the projection operator when projecting onto the space spanned by $m_{1}(\theta )$ and by  $Q_{1}^{\bot}$ the residual operator. 
 The following proposition specifies that the additional component $m_{2}(\theta)$ has to be less than perfectly correlated with $m_{1}$ in order to put more restrictions on the wedge $U$. 
 \begin{lem} Let $\mathcal{V}$: $\mathbb{E}(m_{1,t}(\theta)+m_{2,t}(\theta) -{\mathcal{V}}) =  0$ and $U:=\mathbb{E}(Q_{1}\mathcal{V})$. 
 	For $\mathcal{U}= \left\{U\in [\underbar{U},\bar{U}]\right\}$ define  \[\mathcal{U}^{c}=\left\{U\in [\underbar{U},\bar{U}] : \mathbb{E}Q_{1}^{\bot}\left({m}_{2}(U)-\mathcal{V}\right)  = 0 \right\}\]
 	Then,  $\Theta^{\prime }_{I}\subset\Theta_{I} \quad \text{iff} \quad \mathcal{U}^{c}\subset \mathcal{U} $.
 	\label{over}
 \end{lem}  
 \begin{proof} of Lemma \ref{over}\\
 The additive structure in \eqref{eq:decomp} generalizes as follows for $n_{\theta}\geq 1$:
 \[\mathbb{E}(m_{1,t}(\theta)+m_{2,t}(\theta) -{\mathcal{V}}) =  0\]
 where  the two sets of conditional moments of $m_{1,t}(\theta)$ and  $m_{2,t}(\theta)$ satisfy: 
 	\begin{eqnarray*}
 		\mathbb{E}_{t}\underset{n_{\theta}\times1}{m_{1,t}(\theta)} & = & \mathbb{\mathcal{V}}_{1,t}\in[\underbar{\ensuremath{\mathcal{V}_{1}}},\bar{\mathcal{V}}_{1}]\\
 		\mathbb{E}_{t}\underset{n_{\theta}\times1}{m_{2,t}(\theta)} & = & \mathbb{\mathcal{V}}_{2,t}\in[\underbar{\ensuremath{\mathcal{V}_{2}}},\bar{\mathcal{V}}_{2}]
 	\end{eqnarray*} and hence 
 	 $\mathcal{V}\in\left[\underbar{\ensuremath{\mathcal{V}}}_{1}+
 \underbar{\ensuremath{\mathcal{V}}}_{2},\quad \bar{\mathcal{V}}_{1}
 +\bar{\mathcal{V}}_{2}\right]$.
These conditions imply the following unconditional moment equalities : 
 	\begin{eqnarray*}
 		\mathbb{E}((m_{1,t}(\theta)-\mathbb{\mathcal{V}}_{1,t})\otimes \phi_{t}) & = & 0\\
 		\mathbb{E}((m_{2,t}(\theta)-\mathbb{\mathcal{V}}_{2,t})\otimes \phi_{t}) & = & 0
 	\end{eqnarray*}
 	\normalsize
  Again, to economize on notation, redefine the moments after multiplication with the instrument $\phi$ and drop dependence on $\theta$ and on $t$.     Since $(m_{1},m_{2})$ are 
 	correlated in population, re-project the sum onto the space spanned by ${m}_{1}$.
 	Since the original sum satisfies the moment condition, then the two
 	orthogonal complements will also satisfy it: 
 	\begin{eqnarray*}	
 		\mathbb{E}\left(Q_{1}\left({m}_{1}+{m}_{2}-\mathcal{V}\right)\right) & = & \mathbb{E}\left({m}_{1}+Q_{1}\left({m}_{2}-\mathcal{V}\right)\right)=0\\
 		\mathbb{E}\left(Q_{1}^{\bot}\left({m}_{1}+{m}_{2}-\mathcal{V}\right)\right) & = & \mathbb{E}\left(Q_{1}^{\bot}\left({m}_{2}-\mathcal{V}\right)\right)=0
 	\end{eqnarray*}

 	As in Lemma \ref{nr}, the first
 	set of restrictions identifies a one to one mapping from $U$ to $\Theta_{I}$ as ${m}_{1}+Q_{1}{m}_{2}$ is a function of $m_{1}$, and therefore $\theta^{*}(U)\equiv \mathcal{G}^{-1}(U)$. Plugging this in the second set of restrictions eliminates dependence on  $m_{1}$ and imposes further restrictions on the domain of variation of  $U$. The admissible set for $U$ is now characterized by \[\mathcal{U}^{c}:=\left\{U \in [\underbar{U},\bar{U}]: \mathbb{E}\left(Q_{1}^{\bot}\left({m}_{2}(U)-\mathcal{V}\right)\right)=0\right\}\] Therefore,  $\exists \theta\in \theta(U): \theta\notin \Theta_{I}(U^{c})$ and consequently $\Theta_{I}'\subset \Theta_{I}$.
 \end{proof}

\normalsize

 \begin{proof}  \label{over_surv_proof}
 	\textbf{of Proposition \ref{over_surv}} \\ 
 	(1):  When $Var(B_{t})>0$, the components $m_{1,t}$ and $m_{2,t}$ in the proof of Lemma \ref{over} are such that $Corr(\lambda^{un}_{t},\lambda^{un}_{t}(B^{-1}_{t}-1))\neq 1$, element-wise.  By Lemma \ref{over} we conclude.\newline
 	(2): Suppose that $\Theta_{I}$ is a singleton. Then it must be that the RHS of \eqref{mom_comp_gmm} is zero: Either $B_{t}=0$ for all $t$, or $p(s_{i,t}|S_{t})$ has unit mass on one agent who is unconstrained, and thus $B_{t}=0$.
 \end{proof}

 \subsection{Wealth Distribution with Point mass at the Borrowing Limit.}$ $
 \begin{figure}[H]
 	\begin{centering}			
 		\includegraphics[scale=0.4]{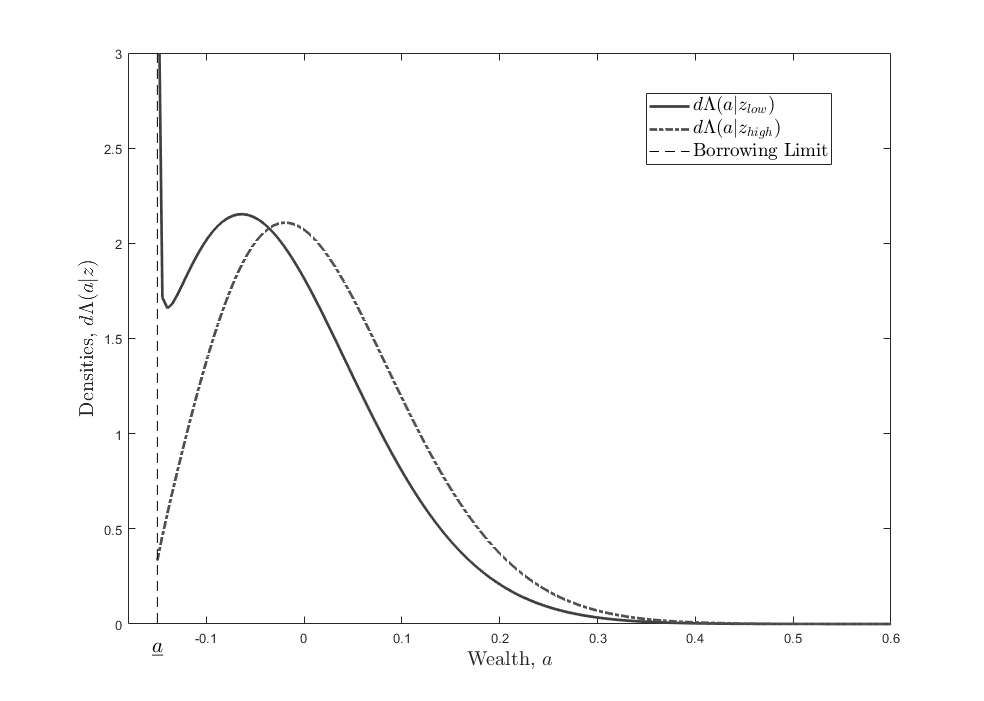}
 		\par
 	\end{centering}   	\protect\caption{Huggett model with a borrowing constraint for wealth $\alpha$ at $\underline{\alpha}<0$, where $\Lambda(\alpha|z_{low})$ is the distribution of the low income type which has a Dirac point mass at $\underline{\alpha}$ (\citet{10.1093/restud/rdab002}).}
 	\label{Achdou} \vspace{-0.1 in}
 \end{figure}
\newpage
\subsection{Monte Carlo Evidence ($500$ replications, $T=200$) } $ $ \label{MCev}

This section presents complementary Monte Carlo evidence on the informativeness of the proportion of borrowing constrained under alternative assumptions about frequency mismatch in the former and presence of measurement error in the latter.  The figures below plot the histogram of the length of the confidence sets (CS) of each parameter in the Krusell-Smith model when the extensive margin is and is not utilized. For classical measurement error in $\hat{B}_{t}$, I add normally distributed noise with standard deviation equal to $10\%$ (Figure \ref{me1}) and $20\%$ (Figure \ref{me2}) of the true standard deviation. Figure \ref{mf1} accommodates for the case in which the consumption cross section is not observed at the same frequency as the aggregate time series i.e. one observation of the cross section for every eight observations of the time series. The missing values are imputed using shape-preserving interpolation. Figure \ref{me1mf1} presents the results for the case that combines both measurement error and missing values.  

\begin{figure}[H]
	\begin{center}			
		\includegraphics[scale=0.5]{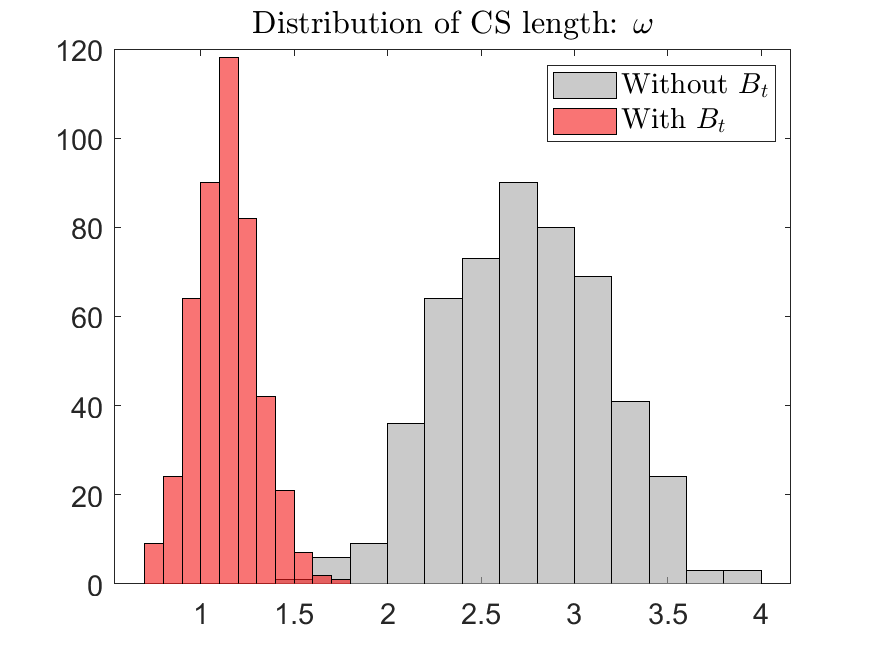}
			\includegraphics[scale=0.5]{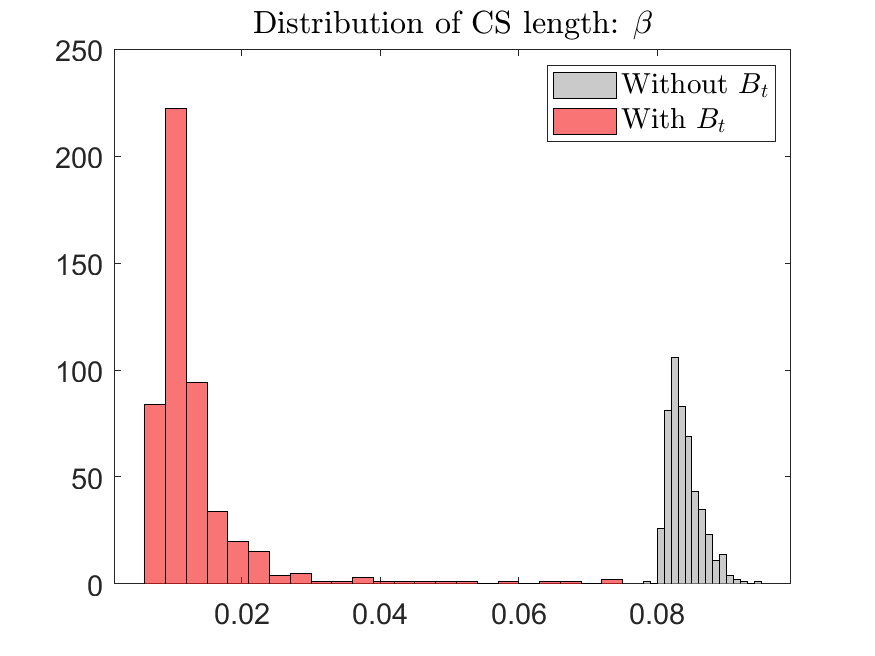}\caption{Benchmark}\label{Bench}
					\end{center}
		
	\end{figure}\begin{figure}[H]
	\begin{center}
		\includegraphics[scale=0.48]{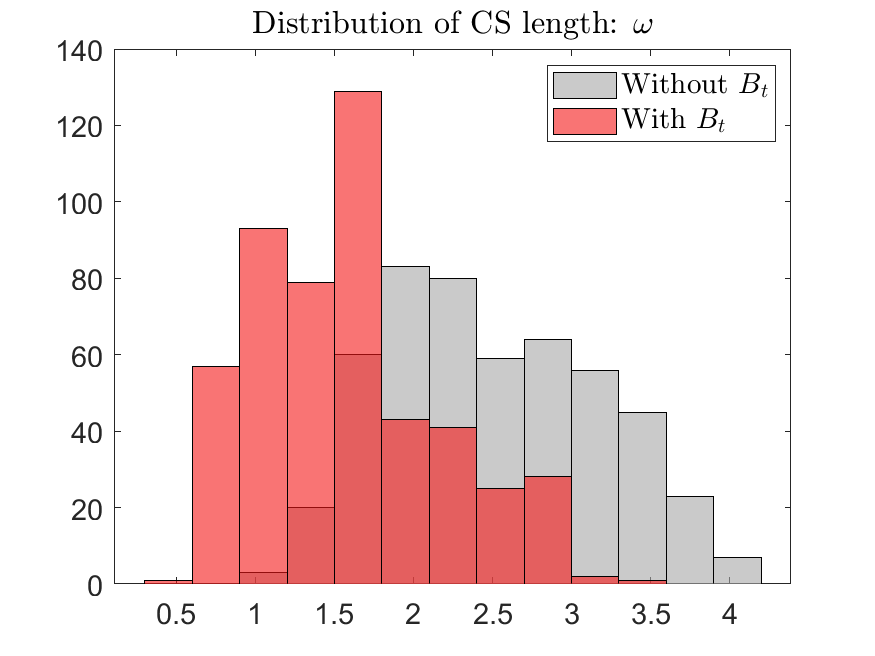}		
		\includegraphics[scale=0.48]{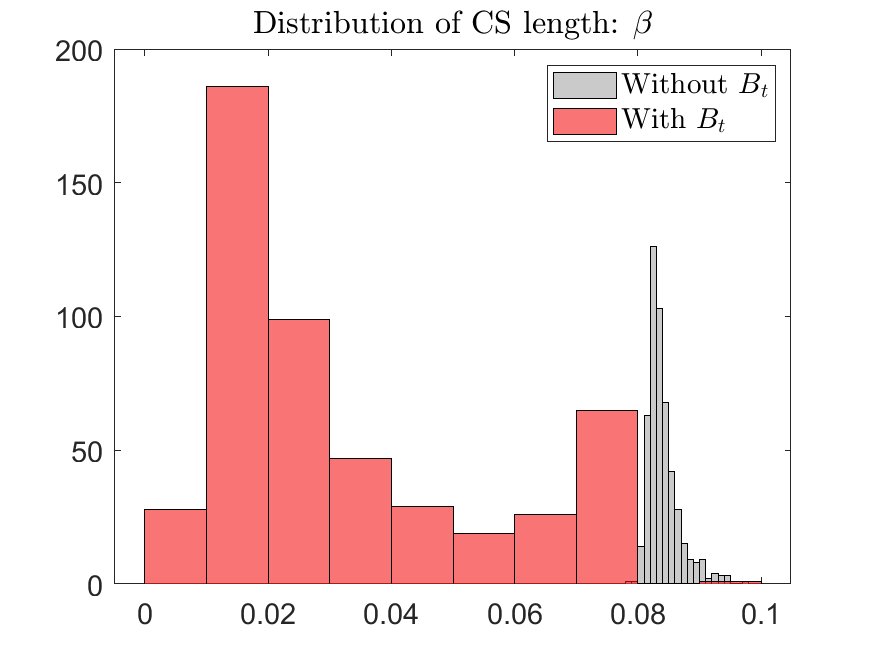}\caption{Measurement error on $B_{t}: \hat{B}_{t}= B_{t} + 0.1\sigma(B_{t})N(0,1)$} \label{me1}
				\end{center}
	
\end{figure}\begin{figure}[H]
\begin{center}
			\includegraphics[scale=0.5]{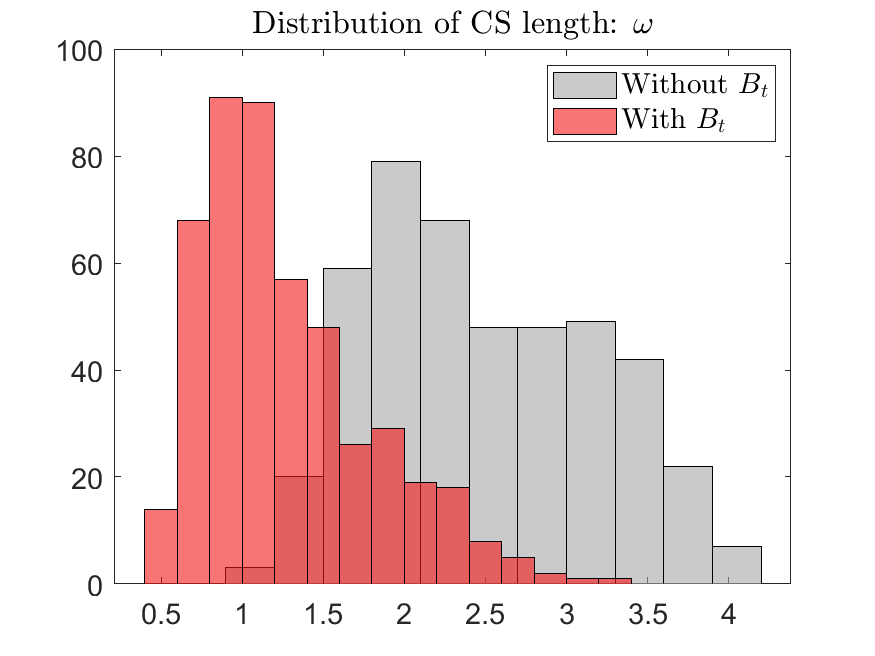}
				\includegraphics[scale=0.5]{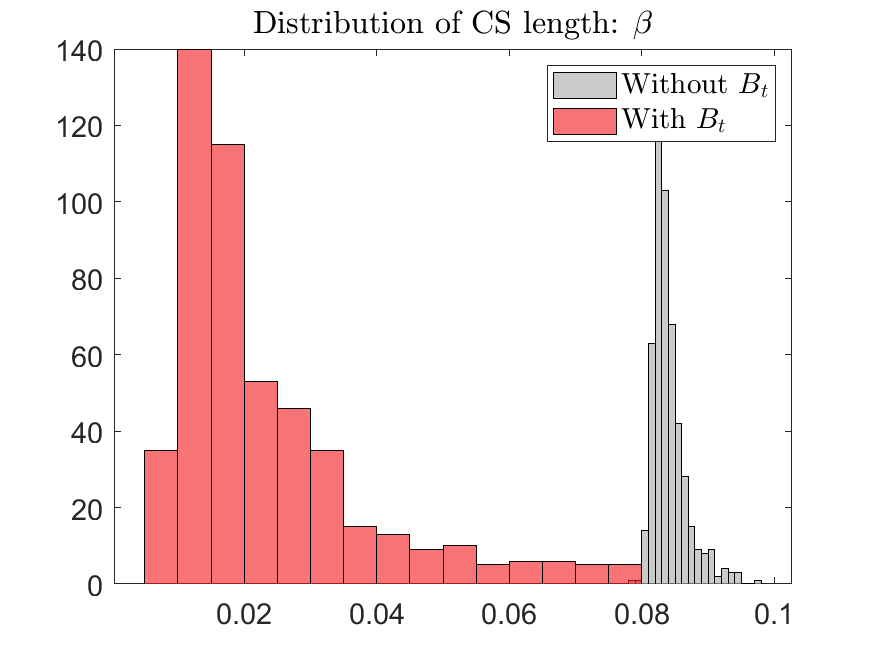} \label{me2}
				\caption{Measurement error on $B_{t}: \hat{B}_{t}= B_{t} + 0.2\sigma(B_{t})N(0,1)$}
						\par
		\end{center}

\end{figure}

\begin{figure}[H]
	\begin{center}			
			\includegraphics[scale=0.5]{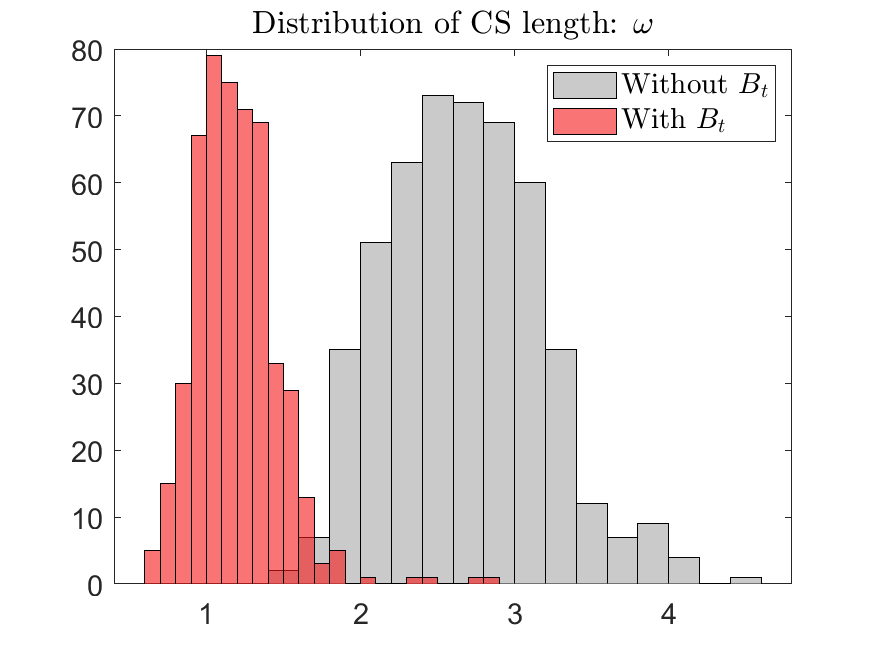}
		\includegraphics[scale=0.5]{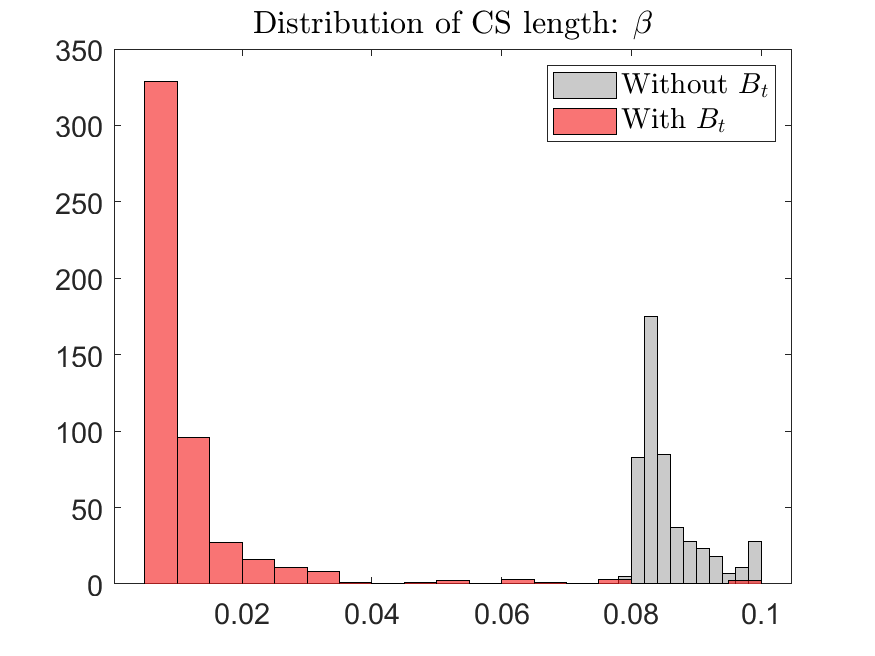}
		\caption{Observing $\{c_{j,t}\}_{j=1..T}$ every eight quarters} \label{mf1}
			\end{center}
	
\end{figure}
		
		\begin{figure}[H]
			\begin{center}
		\includegraphics[scale=0.5]{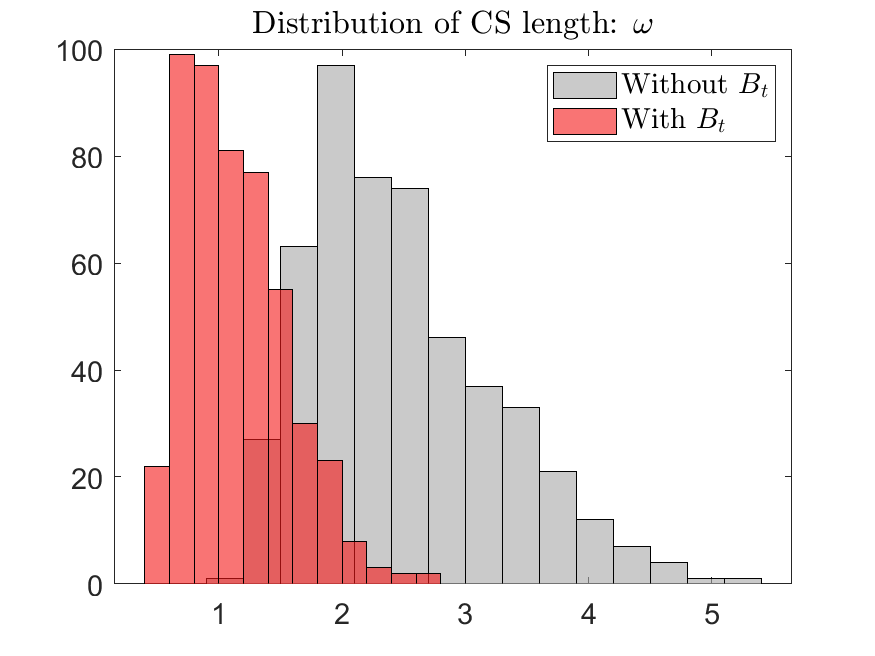}
		\includegraphics[scale=0.5]{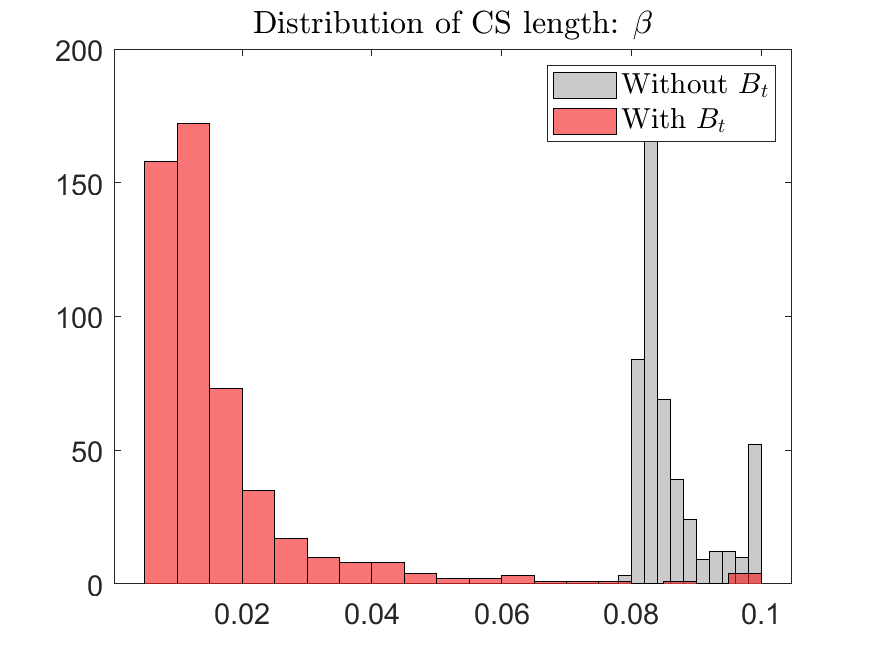}
			\par
	\caption{Measurement error on $B_{t}: \hat{B}_{t}= B_{t} + 0.1\sigma(B_{t})N(0,1)$ and Observing $\{c_{j,t}\}_{j=1..T}$ every eight quarters}  \label{me1mf1}
\end{center}
	
\end{figure}
\normalsize
\subsection{Imputing $B_{t}$ when unobserved} \label{measurements}
In order to obtain a quarterly measure of $B_{t}$, I utilize monthly data from the Business and Consumer Survey (BCS) of the European Commission \citet{BCS}. In particular, the survey asks households about their financial situation, and the possible answers range from "We are saving a lot" to "We are running into debt".\footnote{The exact wording is: "Which of these statements best describes the current financial situation of your household?", and the list of possible answers is: (1) Saving a lot, (2) Saving a little, (3) Just managing to make ends meet on our income, (4) Having to draw on our savings, (5) Running into debt.} In Figure \ref{mixedLCC}, I plot the proportion that corresponds to the latter category (those that claim to be getting indebted), which  increased from $5\%$ in 2008 to $10\%$ by the end of 2016.  

To infer the true proportion of constrained households for the periods in which it is not observed, so all periods between the triennial measures from SHF, I exploit the variation in the BCS data by employing a mixed frequency model, where the quarterly observations on consumers running into debt are linked to the triennial exact measure.   More particularly, let  $B_{t}$ be the proportion of borrowing constrained agents within the quarter and $\Pi_{t}$ be the proportion of consumers that are running into debt. The former are necessarily a share of the latter, as a fraction of the measure of households running into debt are those  that would have liked to take on more debt within the quarter but they were denied by the credit institution, in total or partially. In other words, $B_{t}$ is the share of $\Pi_{t}$ that hits their borrowing limit and are therefore not on their individual Euler equation.\footnote{The households that would select "Running into debt" as a reply in the survey are most likely those that desire to borrow within the period. To relate this to a model, consider the typical budget constraint in the consumption-savings problem with a borrowing constraint at $\underline{a}$:
	\[a_{i,t+1}-a_{i,t} = ra_{i,t} + y_{i,t} - c_{i,t},\quad a_{i,t+1}\geq -\underline{a}\] The survey question identifies the mass of households on  $a_{i,t+1}$ given $a_{i,t}$, where categories (1)-(2) correspond to $\Delta a_{i,t+1}>0$, (3) to  $\Delta a_{i,t+1}=0$ and (4)-(5) on $\Delta a_{i,t+1}<0$. Category 4 identifies the households that need to run down their savings and category 5 identifies the households that need to borrow. A share of the latter will not borrow as much as they would like due to the limit.  The same logic applies to the special case in which $\underline{a}=0$ as long as being indebted includes the boundary $a_{i,t+1}= 0$, but this case is fairly unrealistic.	
	See also Figure \ref{Achdou}  for an example of a wealth distribution based on the \citet{HUGGETT1993953} model with a limit on credit balance, where the distribution for low income types has a point mass at this limit. This is necessarily a fraction of the total mass of households that get indebted within the period and also a fraction of the mass of households that are indebted in general. Whether respondents interpret category 5 in one way or the other does not affect the way it is utilized in the econometric specification.} 
This leads to the following measurement equation: 
\[B_{t}=\zeta_{t}\Pi_{t}\] where $\zeta_{t}\in(0,1)$. 
As shown in Section \ref{meas}, this measurement equation is part of a state space system where $\tilde{B}_{t}$ and $\tilde{\zeta}_{t}$ are the latent states and $\hat{\Pi}^{o}_{t}$ and  $\hat{B}^{o}_{t}$ are corresponding observables at quarterly and triennial frequency respectively. It can be shown that when the exact measure was observed in period $t-1$ $(\tilde{B}^{o}_{t-1})$ but missing in $t$, the (Kalman) filtered measure $(\hat{\tilde{B}}_{t})$ is equal to   
\[\hat{\tilde{B}}_{t} = \rho \tilde{B}_{t-1} +  \frac{\sigma^{2}_{\zeta}+\sigma^{2}_{\Pi}}{2\sigma^{2}_{\zeta}+\sigma^{2}_{\Pi}}\left(v_{\tilde{B}_{t}}-v_{\tilde{\zeta}_{t}}\right)= \rho \tilde{B}_{t-1} +  \frac{\frac{\sigma^{2}_{\zeta}}{\sigma^{2}_{\Pi}}+1}{\frac{2\sigma^{2}_{\zeta}}{\sigma^{2}_{\Pi}}+1}v_{\Pi_{t}}\] \normalsize  
where $(\sigma^{2}_{\zeta},\sigma^{2}_{\Pi})$ are the variances of $\zeta_{t}$ and $\Pi_{t}$ correspondingly, and $\rho$ is the autoregressive parameter for $B_{t}$ in the state equation.

If the proportion $\zeta_{t}$ is not varying much around its long run mean, that is $\sigma^{2}_{\zeta}$ is low, then the mass of constrained agents co-varies with those who are getting indebted. Fluctuations in the latter therefore are very informative for measuring $B_{t}$, and thus the innovation to $\hat{\Pi}^{o}_{t}$ ($v_{\Pi_{t}}$) receives a weight close to one.  Conversely, when $\sigma^{2}_{\zeta}$ is high enough, the variation in the mass of agents that are getting indebted is not as informative for $\tilde{B}_{t}$, and thus the innovation in $\hat{\Pi}^{o}_{t}$ receives a low weight. In this case the filtered measure is then largely based on a forecast from last period's observed proportion of constrained consumers.  

The MLE estimate of $\frac{\sigma^{2}_{\zeta}}{\sigma^{2}_{\Pi}}$ is $0.0425$, yielding a gain parameter equal to $0.9608$. Variation in $\Pi_{t}$ is then informative, albeit not perfectly. The extracted measure is plotted in Figure \ref{mixedLCC}, together with $\Pi^{o}_{t}$ and $B^{o}_{t}$. The mass of the observed credit constrained households (SHF) is lower than the mass of indebted households (BCS), which is in line with the model specification, with the gap between the two increasing after 2009.
 \begin{figure}[H]
	\begin{centering}			
		\includegraphics[scale=0.52]{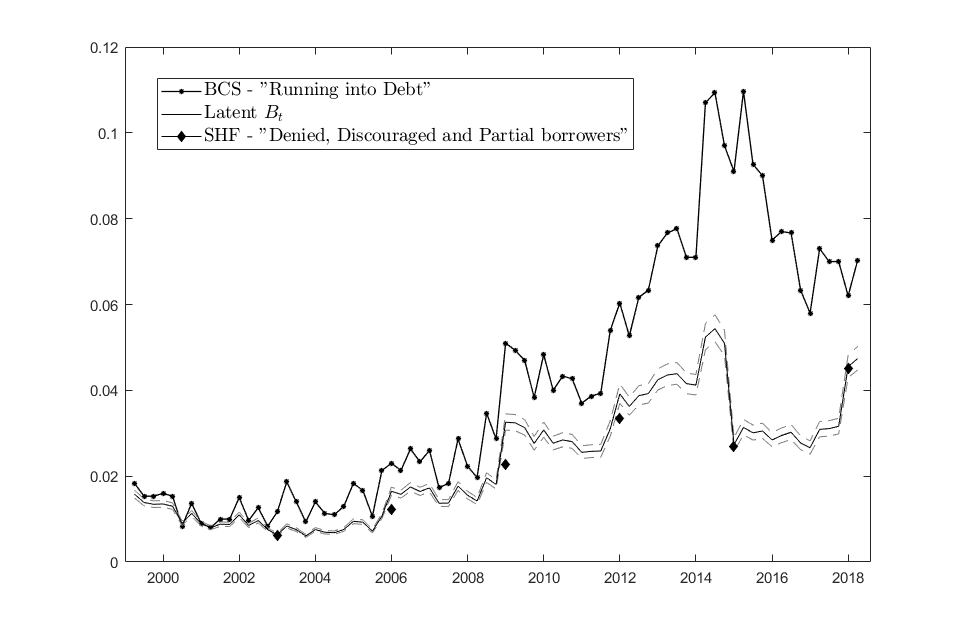}
		\par
	\end{centering}   	\protect\caption{The latent $B_{t}$ is extracted using the SHF measure and exploiting variation in the BCS data.}
	\label{mixedLCC} \vspace{-0.1 in}
\end{figure}

 \subsubsection{Mixed Frequency Model for Quarterly $B_{t}$} \label{mixed}$ $
 \label{meas}
 Using  $\tilde{B}_{t} = \tilde{\Pi}^{o}_{t}+\tilde{\zeta}_{t}$ as a measurement equation, the true (log) quarterly proportion of liquidity constrained consumers is extracted using the following mixed frequency Gaussian linear state space model, where $t=4(j-1)+q$, $t$ is the quarterly observation at year $j$ and $q=\{1,2,3,4\}$ the within year quarter index:\newline
 State Equation (ignoring identities): 
 \begin{eqnarray*}
 	\left( \begin{array}{c}
 		\tilde{b}_{4(j)}\\
 		\tilde{\zeta}_{4(j)}\\
 	\end{array}\right) &=&  \left(\begin{array}{cc}
 		\rho_{b} & 0 \\
 		0 & \rho_{\zeta} 
 	\end{array}\right)  \left( \begin{array}{c}
 		\tilde{b}_{4(j-1)+3}\\
 		\tilde{\zeta}_{4(j-1)+3}\\
 	\end{array}\right) + \left(\begin{array}{c}
 		\nu_{\tilde{b},4(j)}\\
 		\nu_{\tilde{\zeta},4(j)}
 	\end{array}\right)\end{eqnarray*}
 Observation equation:
 \begin{eqnarray*}
 	\left(\begin{array}{c}
 		\tilde{\Pi}^{o}_{4(j)}\\
 		\tilde{B}^{o}_{j}\end{array}\right)&=& \left( \begin{array}{cccccccccccc}
 		1 & 0 & 0 &0 & 0 & 0 &0 & 0 &1& 0 &... & 0\\
 		\frac{1}{8} &\frac{1}{8} & \frac{1}{8} &\frac{1}{8} & \frac{1}{8} & \frac{1}{8} &\frac{1}{8} &\frac{1}{8}&0&0&...&0 \\
 	\end{array} \right)\left( \begin{array}{c}
 		\tilde{b}_{4(j)}\\
 		\tilde{b}_{4(j-1)+3}\\
 		\tilde{b}_{4(j-1)+2}\\
 		\tilde{b}_{4(j-1)+1}\\
 		\tilde{b}_{4(j-1)}\\
 		\tilde{b}_{4(j-2)+3}\\
 		\tilde{b}_{4(j-2)+2}\\
 		\tilde{b}_{4(j-2)+1}\\
 		\tilde{\zeta}_{4(j-1)+3}\\
 		\tilde{\zeta}_{4(j-1)+2}\\
 		\tilde{\zeta}_{4(j-1)+1}\\
 		\tilde{\zeta}_{4(j-2)}\\
 		\tilde{\zeta}_{4(j-2)+3}\\
 		\tilde{\zeta}_{4(j-2)+2}\\
 		\tilde{\zeta}_{4(j-2)+1}\\
 		\tilde{\zeta}_{4(j-2)}
 	\end{array}\right) + v_{4(j)} 
 \end{eqnarray*} 
 where  $(\nu_{\tilde{b},t},\nu_{\tilde{\zeta},t})\sim N(0,diag(\Sigma_{\nu}))$ and $v_{t}\sim N(0,diag(\Sigma_{v}))$. The last diagonal component of $\Sigma_{v}$ is calibrated to the standard error from the estimation of $B_{t}$ in the SHF survey. The rest of the components of the diagonal are calibrated to $1\%$ of the variance of the BCS measure ($\tilde{\Pi}_{t}$).

\newpage
 \subsection{Data} \vspace{-0.2 in}
 \begin{flushleft} 	
 	\begin{table}[H] 
 		\begin{centering}
 			\caption{Sources and Transformations}\vspace{0.3 in}
 			\label{table:data}
 			\resizebox{\columnwidth}{!}{
 				\begin{tabular}{c|c|c|c|}
 					\hline
 					Variable& Data ($1999Q1-2018Q1$, excl. IBEX35, starting $2001Q1$.)& Source & Final Transformation \\ \hline 
 					$C_{t}$& Final Consumption Expenditures, S.A ($C^{tot}$) & FRED (OECD) & $\frac{C^{tot}}{Pop*P}*100$\\ 				$Pop_{t}$ & Active Population: Aged 15-74, S.A.& FRED (OECD)& \\
 					$L_{t}$ & Total Hours Worked , S.A. ($L^{tot}$)& Datastream (National Accounts) & $\frac{H^{tot}}{Pop}$\\ 
 					$P_{t}$ & GDP Implicit Price Deflator& FRED (OECD) &\\
 					$\pi_{t}$ & --& -- & $ln(\frac{P_{t}}{P_{t-1}})$\\
 					$W_{t}$ & Total Employee Compensation , S.A. ($W^{tot}$)&  Thomson Reuters (OECD) & $\frac{W^{tot}}{H}$\\
 					$R_{t+1}$ & EONIA (Quarterly rate, $i_{t}$)&  European Central Bank & $ i_{t} - \pi_{t+1} $\\
 					$R_{g,t+1}$ & Treasury Bills (Quarterly rate, $i_{g,t}$)&  FRED & $ i_{g,t} - \pi_{t+1} $\\
 					$R_{ib,t+1}$ & 3-month Interbank loans  (Quarterly rate, $i_{ib,t}$)&  FRED & $ i_{ib,t} -\pi_{t+1} $\\
 					$R_{e,t+1}$ & IBEX35/IGBM Indices (Quarterly Total return, $i_{e,t}$ ))&  Thomson Reuters & $ i_{e,t} -\pi_{t+1} $\\
 					$(c_{i,t})$& Consumption Expenditures& Survey of Household Finances& See 7.7.1\\				
 					\hline
 			\end{tabular}}
 			\par\end{centering}  
 		
 	\end{table}
 	Excluding durables from aggregate consumption does not seem a reasonable thing to do since we consider liquidity constraints, which are more binding in the case of large purchases e.g. automobiles. As far as housing is concerned, aggregate consumption measures include imputed rents, which can be thought as the service value from housing.  
 \end{flushleft}
 \subsubsection{Imputing Consumption using the Survey of Household Finances}\label{cons_sh}
 Consistent with measurement of the aggregate, annual nominal consumption is computed using survey weights as follows: \begin{eqnarray*}
 	c^{nom,annual}_{i,t} &=& 12*(\textit{Imputed Monthly Rent} [code: 2.31]\\
 	&& + \textit{Car Purchases}/12 [code: 2.74] + \textit{House Durables}/12 [code:p2.70]\\
 	&& +\textit{monthly non durable consumption} [code:p9.1])\end{eqnarray*} 
 Correspondingly, annual consumption is distributed equally to quarterly consumption. 
\newpage
 \subsubsection{Interest Rates} Ex-Post Real interest rates:
 
 \begin{figure}[H]
 	\begin{centering}
 		\includegraphics[scale=0.55]{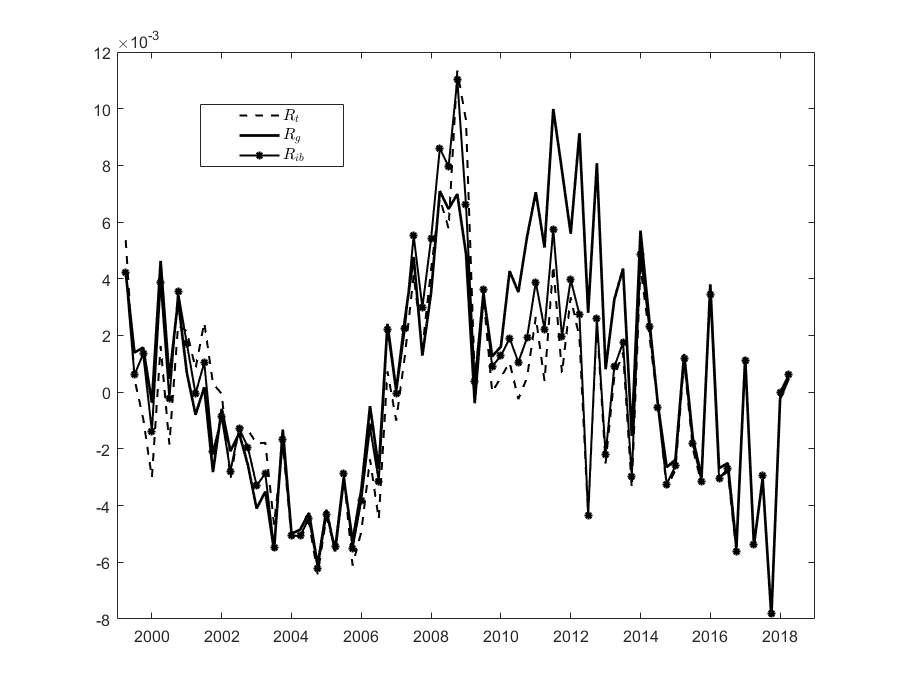} \label{rates}
 		\par
 	\end{centering}  
 	\protect\caption{$ln(R_{t+1}) := ln(1+i_{t}) - ln(1+\pi_{t+1})$}
 	\label{rates}
 \end{figure}

\newpage
 \renewcommand{\theequation}{B.\arabic{equation}}
 
 \setcounter{equation}{0}
 \section{Appendix B-I (Online)}  \label{AppB} 
 
\subsection{Analytical Example on Identification using Micro Data}$ $\\
Recall that conditions in \eqref{eq:genzeldes} and instrument $z_{i,t}(\geq 0)$ generate 
the following unconditional moment inequality, which  does not pin down a unique vector for $(\beta,\omega)$ : 
\begin{eqnarray}
\mathbb{E}\left[(U'_{1}(c_{i,t};l_{i,t},\omega,S_{t})-\beta(1+r^{e}_{i,t+1})U'_{1}(c_{i,t+1};l_{i,t+1},\omega, S_{t+1})) z_{i,t}\right] & \geq & 0 
\end{eqnarray}  
 In order to derive analytical identification regions I simplify the model by assuming exogenous income (so labor supply is not relevant), a riskless interest rate $r^{e}_{i,t+1}=r$, and eliminate $S_{t}$ from the utility function. I adopt the approximation of \citet{Hall1978}\footnote{Please see proof of Corollary 5 in p.987.} using CRRA utility $U(c_{i,t};\gamma):= \frac{c_{i,t}^{1-\gamma }-1}{1-\gamma }$, which implies
the following law of motion for consumption: \vspace{-0.1 in}
\begin{eqnarray}
c_{i,t+1}&=&\rho c_{i,t}+\tilde{\lambda}_{i,t}+\phi_{i,t} +\epsilon_{i,t+1}\label{eq:HallZeldes}
\end{eqnarray} 
where $\tilde{\lambda}_{i,t}\equiv-(U''(c_{i,t};\gamma))^{-1}\lambda_{i,t}$, $U''(c_{i,t};\gamma)$ is the second derivative of the utility function, $\rho=\left(\beta(1+r)\right)^{-\frac{U'(c_{i,t};\gamma)}{c_{i,t}U''(c_{i,t};\gamma)}}=\left(\beta(1+r)\right)^{\frac{1}{\gamma}}$ and $\phi_{i,t}$ is the residual of this approximation which is also set to zero.\footnote{In a recent paper \citet{commault2019} challenged \citet{Hall1978}'s random walk result for consumption in the case of isoelastic utility. In Section \ref{alt} in Appendix B-II, I  clarify why the identification analysis is unaffected by this simplification; in a nutshell, $\phi_{i,t}$ has very similar properties to $\tilde{\lambda}_{i,t}$.} Moreover, $\epsilon_{i,t+1}$ is the rational forecast error such that $\mathbb{E}_{t}\epsilon_{i,t+1}=0$ and summarizes the uncertainty faced by the otherwise rational individual.  Equivalently, consumption growth is equal to   \[\Delta c_{i,t+1}=(\rho-1)c_{i,t}+\tilde{\lambda}_{i,t}+\epsilon_{i,t+1}\] Using $z_{i,t}$ as an instrument and computing its covariance with $\Delta c_{i,t+1}$ leads to 
 \[Cov(z_{i,t},\Delta c_{i,t+1})=(\rho-1)Cov(z_{i,t},c_{i,t})+Cov(z_{i,t},\tilde{\lambda}_{i,t})\] Solving for $\rho$ and using that  higher income relaxes the budget constraint, that is,  $Cov(z_{i,t},\tilde{\lambda}_{i,t})$ is negative, the set of admissible values for $\rho$ is as follows: 
\begin{eqnarray}
\rho_{{ID,1}}&:=& \left(0,1+\frac{Cov(z_{i,t},\Delta c_{i,t+1})}{Cov(z_{i,t}c_{i,t})}\right] = \left(0,1+\rho_{IV}\right] \label{eq:bound1}
\end{eqnarray} where $\rho_{IV}$ is the IV estimate one would obtain when ignoring liquidity constraints.
Avoiding to take a stance on the exact nature of $\tilde{\lambda}_{i,}$, which can involve a vast amount of unobserved information, results in set identification. 
 Nevertheless, we are still able to infer certain facts about the household's preferences and economic behavior and conclusions are valid across alternative environments. 
\subsection{Bounds on Risk Aversion}
Given the set of admissible values for the reduced form coefficient $\rho$, we can recover the implied bounds for the risk aversion parameter, $\gamma_{ID}$, using that $\rho=\left(\beta(1+r)\right)^{\frac{1}{\gamma}}$. Treating $(r,\beta)$ as known, if $\beta(1+r)=1$, risk aversion is unidentified ($\gamma_{ID} =\mathbb{R}_{+}$) as consumption follows a random walk. If $\beta(1+r)\neq 1$, then risk aversion is set identified and has a very intuitive interpretation, as it reflects restrictions on preferences implied by the presence of non-diversifiable income risk. 

In particular, for $\beta(1+r)<1$ the household is impatient and does not accumulate wealth indefinitely. Using \eqref{eq:bound1}, since income and consumption growth are negatively correlated {see e.g. \citet{Deaton}}, risk aversion is bounded by above and the set of values consistent with the data is: \vspace{-0.1 in}\small \[\gamma_{ID,1}= \left\{\gamma \in \mathbb{R}_{+}: \gamma < \frac{\|\log(\beta(1+r))\|}{\log \left(\frac{Cov(z_{i,t},c_{i,t})}{Cov(z_{i,t},c_{i,t})-\|Cov(z_{i,t},\Delta c_{i,t+1})\|}\right)}\right\} \]\normalsize The stronger the negative correlation,  the lower the upper bound on risk aversion, indicating that the less risk averse household is not accumulating enough wealth to fully insure against income risk.


\subsection{\textbf{The extensive margin as additional information}} Suppose now that we observe the dichotomous response of the household over time to a survey question that asks whether the household is (or expects to be in the near future) financially constrained. An honest household will answer positively whenever $\lambda_{i,t}>0$. 

In Section \ref{deriv} I show that had the data been generated by a model with a restriction on borrowing, the consumption growth process becomes as follows:
\begin{eqnarray} 
	&&\Delta c_{i,t+1} \label{eq:cons_growth}\\
	& =& (\rho-1)c_{i,t}+\epsilon_{i,t+1}+\mathbf{1}(\lambda_{i,t}>0)\left(-{\rho}c_{i,t} + \bar{y}\left( 1- \frac{\rho}{1+r}(1-F_{y_{i}}(x^{\star}_{i}))\right)\right)\nonumber\\
&:=& (\rho-1)c_{i,t}+\epsilon_{t+1}+\mathbf{1}(\lambda_{i,t}>0)\left(\lambda_{1}c_{i,t} + \lambda_{0,i}\right)  \nonumber  \end{eqnarray} 
where $F_{y_{i}}(x^{\star}_{i})$ is the the probability of hitting the constraint, which depends on the distribution of income $(y_{i})$ and cash on hand threshold $x^{\star}_{i}$. The higher this probability is, the higher the savings rate out of income, increasing thereby consumption growth. 

The relevant question to ask is whether we can utilize the survey information without knowing the true data generating process i.e. the form of $\lambda_{i,t}$ in \eqref{eq:cons_growth}, and whether this leads to even more precise estimates of risk aversion. 
First, dichotomous responses to the survey can be used to estimate $\mathbb{P}_{t}(\lambda_{i,t}>0)$.
Therefore, for $v\in \mathbb{R}$ and $\epsilon_{i,t+1}$ with cumulative density $\Phi_{i,t}(.)$, the distribution function of consumption growth is: \footnotesize
\begin{eqnarray} 
\mathbb{P}_{t}(\Delta c_{i,t+1}< v) & =& \Phi_{i,t}(v -({\rho}-1)c_{i,t})\mathbb{P}_{t}(\lambda_{i,t}=0)+\Phi_{i,t}(v -({\rho}-1)c_{i,t}-\lambda_{i,t})\mathbb{P}_{t}(\lambda_{i,t}>0) \nonumber \\
&\leq&\Phi_{i,t}(v -({\rho}-1)c_{i,t})\mathbb{P}_{t}(\lambda_{i,t}=0)+\Phi_{i,t}(v -\tilde{\rho}_{{IV}}c_{i,t})\mathbb{P}_{t}(\lambda_{i,t}>0)  \label{eq:probit4}
\end{eqnarray} \normalsize
where the second line uses the only information available: that is $\lambda_{i,t}>0$, and that the IV regression estimate of  ${\rho}-1$ in \eqref{eq:cons_growth}, $\tilde{\rho}_{{IV}}$, is a lower bound to the true value.\footnote{To see this, notice that  $\tilde{\rho}_{{IV}}={\rho}-1 + {}{{\left({\mathbb{C}ov(c_{i,t},z_{i,t})^{-1}\mathbb{C}ov(z_{i,t},\lambda_{i,t}\mid \lambda_{i,t}>0)}{}\mathbb{P}(\lambda_{i,t}>0)\right)}}  < {\rho}-1 
	$. The bias term is negative as higher income relaxes the constraint and can be verified by the solution of the model, as $\lambda_{1}=-\rho$.}

The simplest way to show that \eqref{eq:probit4} provides additional information is to show that there exists parameter value in  $\rho\in\rho_{{ID,1}}$ that is not consistent with \eqref{eq:probit4}. 
It can be shown that $\tilde{\rho}_{{IV}}$, which is consistent with \eqref{eq:bound1}, does not satisfy \eqref{eq:probit4} when the probability of being constrained is higher than a certain threshold  which depends on the true model:\footnote{ The result is obtained as follows:
	Let $p_{i,t}:=\mathbb{P}_{t}(\lambda_{i,t}>0)$. If $\tilde{\rho}_{{IV}}$ is admissible, evaluating the LHS of the inequality in \eqref{eq:probit4} at the true $\tilde{\rho}=\rho-1$ and the RHS at $\tilde{\rho}_{{IV}}$, and solving for $p_{i,t}$ yields:  \begin{eqnarray*}
		p_{i,t}&\leq& \frac{\Phi_{i,t}(v-\tilde{\rho} c_{i,t})-\Phi_{i,t}(v-\tilde{\rho}_{IV}c_{i,t})}{\Phi_{i,t}(v-\tilde{\rho}c_{i,t})-\Phi_{i,t}(v-\lambda_{0}-(\tilde{\rho}+\lambda_{1})c_{i,t})}:=g(p_{i,t})
	\end{eqnarray*}} \begin{eqnarray}\mathbb{P}_{t}(\lambda_{i,t}>0)\in\left(g_{i,t},1\right) \label{eq:Probbound}\end{eqnarray}  where  \[g_{i,t}:=\frac{\Phi_{i,t}(v-({\rho}-1)c_{i,t})-\Phi_{i,t}(v-\tilde{\rho}_{IV}c_{i,t})}{\Phi_{i,t}(v-({\rho}-1)c_{i,t})-\Phi_{i,t}(v-\lambda_{0,i}-({\rho}-1+\lambda_{1})c_{i,t})}\] 
  
For higher values of the intensive distortion $\lambda_{1}$, e.g. $2\lambda_{1}$ the admissible range (gray shade) becomes larger.   The more severe the distortions, the more likely is that additional information matters.   


    \begin{figure}[H]
  	\begin{center}
  		\includegraphics[scale=0.25]{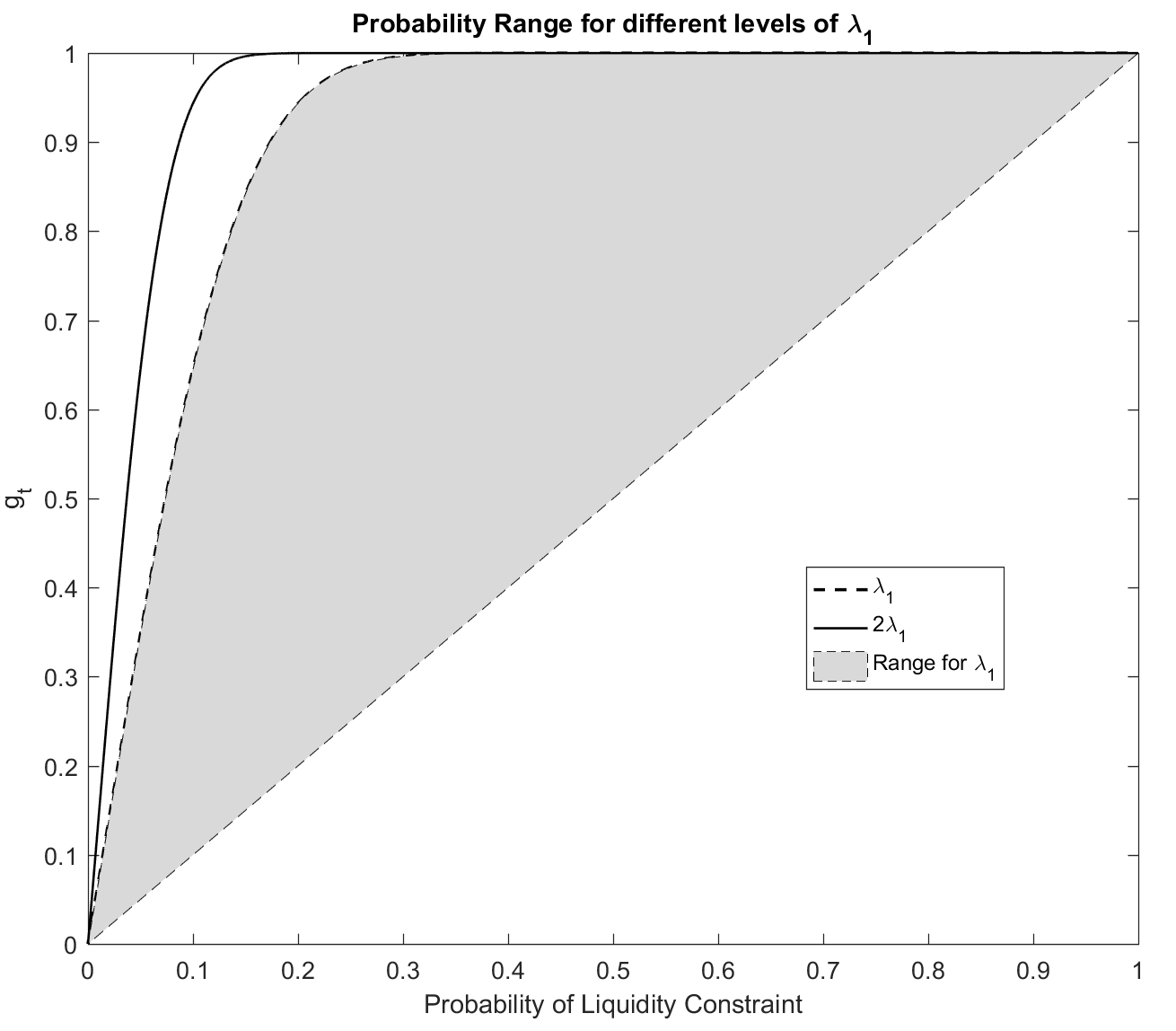} 
  		\protect\caption{\small Range for $v=0,c_{i,t}=1,r=0.05,\sigma_{\epsilon}=0.1,\bar{y}=0$}
  	\end{center} \label{refinement}
  \end{figure}

\section{Appendix B-II (Online)}
This Appendix contains derivations relating  to Appendix B-I.
\subsection{Deriving Consumption Growth} \label{deriv}
Consider the approximate Euler equation that corresponds to \ref{eq:HallZeldes}, where we set $\phi_{i,t}\approx 0$. Moreover, for simplicity, assume that  $y_{i,t}$ is exogenous and iid, with $\mathbb{E}_{t}y_{i,t+1}=\bar{y}$. The household can only save in risk free government bonds and cannot borrow:
\begin{eqnarray*}
	\mathbb{E}_{t}c_{i,t+1}&=&\rho c_{i,t}+\tilde{\lambda}_{i,t} \\
	\alpha_{i,t+1}&=& (1+r)\alpha_{i,t}+ y_{i,t}-c_{i,t} \\
	\alpha_{i,t+1}&\geq&	0
\end{eqnarray*} 
When there are no liquidity constraints, the solution can be obtained as follows:  iterating the budget constraint forward and  imposing the transversality condition $\lim_{j\to\infty}(1+r)^{j}a_{i,t+j}=0$, 
\begin{eqnarray*}
	(1+r)a_{i,t} &=& c_{i,t}-y_{i,t}+\sum_{j=1..\infty}(1+r)^{-j}(c_{i,t+j}-y_{i,t+j})
\end{eqnarray*}
Taking conditional expectations and using the approximate Euler equation, we obtain that 
 $c^{pf}_{i,t}=(1+r-\rho)a_{i,t}+\frac{1+r-\rho}{1+r}(y_{i,t}+\frac{\bar{y}}{r})= \frac{(1+r-\rho)}{1+r}(x_{i,t}+\frac{\bar{y}}{r})$, where $x_{i,t}$ is cash on hand. Due to the linarization, the perfect foresight solution coincides with the solution under uncertainty. Nevertheless, risk aversion does have an effect on $\rho$, and hence the slope of the policy function. 
 
 The only source of concavity for the consumption function is the binding liquidity constraint.
 What I am after here is a piece-wise approximation of the policy function in the presence of constraints, $\tilde{\lambda}_{i,t}\neq 0$, which is analytically tractable. I consider one piece of the approximation to be the policy function below the threshold $x^{\star}$ where the agent becomes hand to mouth, and a linear approximation to the policy function above $x^{\star}$. 
  
 When the household is constrained, necessarily, $c^{con}_{i,t}=x_{i,t}$. Hence, when the household  is unconstrained, its consumption function should become less steep as it is now saving some of its cash on hand. Inspecting the perfect foresight solution, for the consumption function to be continuous at $x^{\star}$,     
 unconstrained behavior should be equal to $c^{unc}_{i,t}=\frac{(1+r-\rho)}{1+r}x_{i,t}$. The difference between $c^{unc}$ and $c^{pf}$, $\frac{(1+r-\rho)}{1+r}\frac{\bar{y}}{r}$, should be the precautionary savings due to the presence of the constraint. These savings are  a share of permanent income $\frac{1+r}{r}\bar{y}$. An  example of a numerical solution that verifies this parallel shift is in Figure \ref{conuncon} below.\footnote{See also \citet{10.1257/jep.15.3.23}.}

\begin{figure}[H]
	\begin{centering}
		\includegraphics[scale=0.3]{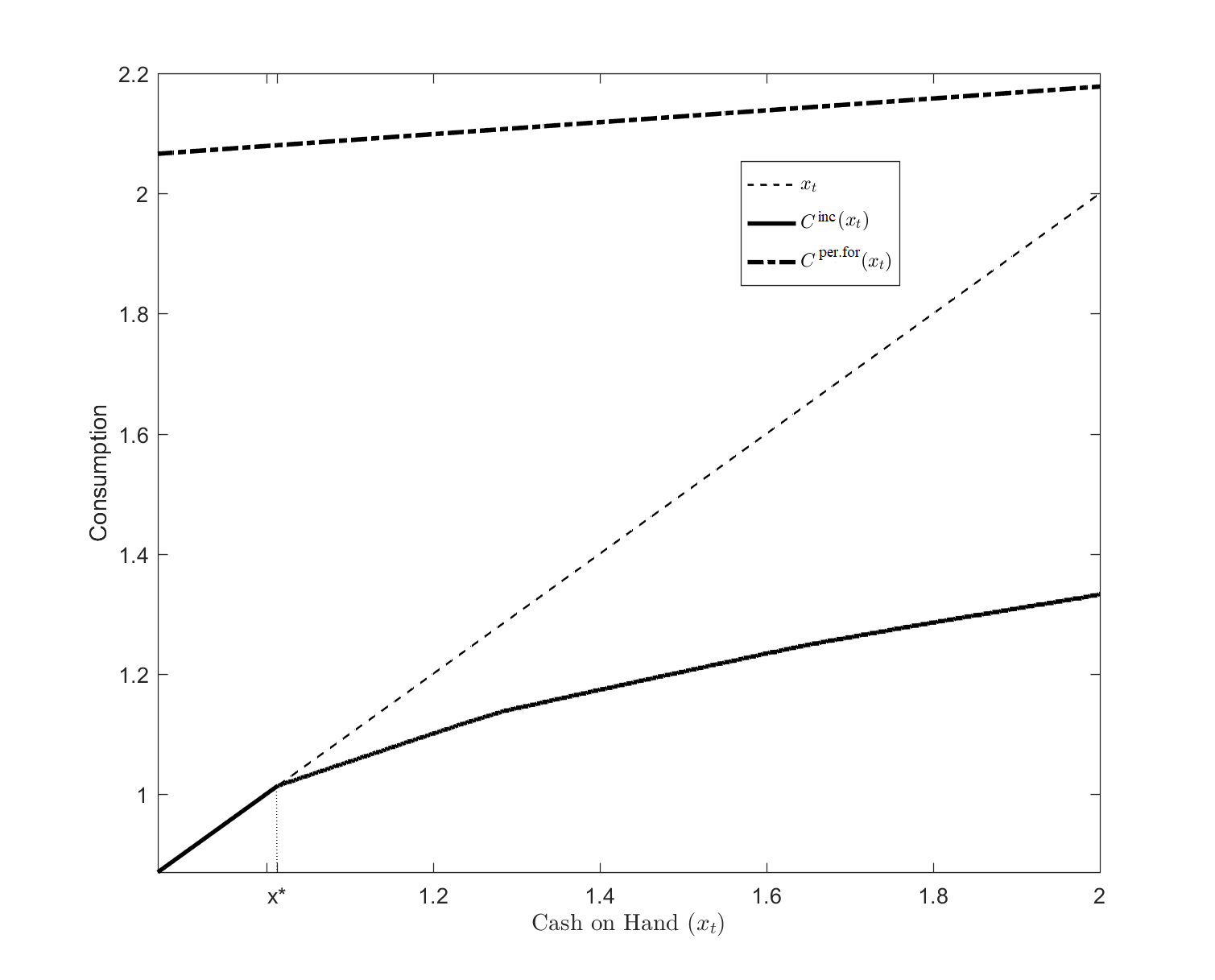} \vspace{-0.1 in}
		\protect\caption{Example of Optimal Consumption: Constrained versus Unconstrained Consumption. Notice also that for large values of $x_{t}$, consumption in incomplete markets with borrowing constraints, $C^{inc}_{t}$, and the perfect foresight solution $C^{per.for}_{t}$ become parallel, as expected. Their difference reflects precautionary savings. }
	\end{centering} \label{conuncon}
\end{figure}

\textbf{Deriving equation \ref{eq:cons_growth}}:
 Using the 		
constrained agent's Euler equation, it has to be the case that $\rho x_{i,t} = \mathbb{E}(c_{i,t+1}|\lambda_{i,t}>0)-\lambda_{i,t}$.  Expectations include the possibilities of being constrained or not in $t+1$, thus:	
\begin{eqnarray*}
	\rho x_{i,t}& = &\mathbb{E}(c_{i,t+1}|\lambda_{i,t}>0)-\lambda_{i,t}\\
	& = &\mathbb{E}(x_{i,t+1}|\lambda_{i,t}>0)\mathbb{P}(\lambda_{i,t+1}>0|\lambda_{i,t}>0)\\
	&  & + \left(1-\frac{\rho}{1+r}\right)\mathbb{E}(x_{i,t+1}|\lambda_{i,t}>0)\mathbb{P}(\lambda_{i,t+1}=0|\lambda_{i,t}>0)-\lambda_{i,t}\\
	& = &\mathbb{E}(x_{i,t+1}|\lambda_{i,t}>0)\left[1-\frac{\rho}{1+r}\mathbb{P}(\lambda_{i,t+1}=0|\lambda_{i,t}>0)\right] -\lambda_{i,t}\\
	& = &\bar{y}\left[1-\frac{\rho }{1+r}\mathbb{P}(\lambda_{i,t+1}=0|\lambda_{i,t}>0)\right] -\lambda_{i,t}\\
	& = &\bar{y}\left[1-\frac{\rho }{1+r}(1-F_{y_{i}}(x^{\star}_{i}))\right] -\lambda_{i,t}
\end{eqnarray*}		
where in the last two lines we use that cash on hand in $t+1$ if the agent is constrained at $t$ is equal to $y_{i,t+1}$ whose expected value is $\bar{y}$, while the probability of being constrained in $t+1$ is $F_{y_{i}}(x^{\star}_{i})$, $x^{\star}$ being the cash on hand threshold. Thus, $\lambda_{i,t}= - \rho x_{i,t} +\bar{y}\left[1-\frac{\rho }{1+r}(1-F_{y_{i}}(x^{\star}_{i}))\right]$ and 
\begin{eqnarray*}\Delta c_{i,t+1}& =& \tilde{\rho} c_{i,t}+\epsilon_{t+1}+\mathbf{1}(\lambda_{i,t}>0)\left(-{\rho}c_{i,t} + \bar{y}\left( 1- \frac{\rho }{1+r}(1-F_{y_{i}}(x^{\star}_{i}))\right)\right)\nonumber\\
	&:=& \tilde{\rho} c_{i,t}+\epsilon_{t+1}+\mathbf{1}(\lambda_{i,t}>0)\left(\lambda_{1}c_{i,t} + \lambda_{0,i}\right)\end{eqnarray*}  
\normalsize
which is equation \ref{eq:cons_growth}.

\subsection{Alternative Derivation to \citet{Hall1978}} \label{alt} The previous analysis ignored the approximation error due to the convexity in marginal utility. I next show that $\lambda_{i,t}$ and $\phi_{i,t}$ have similar properties.  
Following \citet{commault2019}, consider the Euler equation:
\begin{eqnarray*}
	U'\left((\beta(1+r))^{\frac{1}{\gamma}}c_{i,t}\right)&=& \mathbb{E}_{t}U'(c_{i,t+1})+(\beta(1+r))^{-1}\lambda_{i,t}\\
	(\beta(1+r))^{\frac{1}{\gamma}}c_{i,t}&=&U^{'-1}\left(U'(\mathbb{E}_{t}c_{i,t+1}-\phi_{i,t})+(\beta(1+r))^{-1}\lambda_{i,t}\right)
\end{eqnarray*} 
where $\phi_{i,t}>0$ if marginal utility is convex. A first order approximation of the RHS around $\lambda_{i,t}=0$ yields:
	\begin{eqnarray*}
	(\beta(1+r))^{\frac{1}{\gamma}}c_{i,t}&\approx& \mathbb{E}_{t}c_{i,t+1}-\phi_{i,t} - \tilde{\lambda}_{i,t}
\end{eqnarray*} 
where $\tilde{\lambda}_{i,t} = -\frac{(\beta(1+r))^{-1}\lambda_{i,t}}{U''((\beta(1+r))^{\frac{1}{\gamma}}c_{i,t})}>0$. The distortion to random walk comes from both the precautionary savings effect and occasionally binding constraint. Repeating the analysis in \ref{deriv} yields that:\vspace{-0.1 in}	\small
\begin{eqnarray*}
	\lambda_{i,t}+\phi_{i,t}& = &\bar{y}\left[1-\frac{\rho }{1+r}(1-F_{i,y}(x^{\star}_{i}))\right] -\rho x_{i,t}, \quad and
\end{eqnarray*}\vspace{-0.2 in}
\begin{eqnarray*}\Delta c_{i,t+1}& =& \tilde{\rho} c_{i,t}+\epsilon_{t+1}+\mathbf{1}(\lambda_{i,t}>0)\left(-{\rho}c_{i,t} + \bar{y}\left( 1- \frac{\rho }{1+r}(1-F_{y_{i}}(x^{\star}))\right)\right)+ \mathbf{1}(\lambda_{i,t}=0)\phi_{i,t}
\end{eqnarray*} 	\normalsize
Finally, the bound derived in \ref{eq:bound1} is identical because $\phi_{i,t}$ and $Cov(\phi_{i,t},y_{i,t})$  have the same sign as $\lambda_{i,t}$  and $Cov(\lambda_{i,t},y_{i,t})$ respectively.

\end{document}